\documentclass[twocolumn,journal]{IEEEtran}
\usepackage[T1]{fontenc}
\usepackage[latin9]{inputenc}
\usepackage{color}
\usepackage{float}
\usepackage{booktabs}
\usepackage{dsfont}
\usepackage{amsmath}
\usepackage{amsthm}
\usepackage{amssymb}
\usepackage{graphicx}
\usepackage[unicode=true,
 bookmarks=true,bookmarksnumbered=true,bookmarksopen=true,bookmarksopenlevel=1,
 breaklinks=false,pdfborder={0 0 0},pdfborderstyle={},backref=false,colorlinks=false]
 {hyperref}
\hypersetup{pdftitle={Your Title},
 pdfauthor={Your Name},
 pdfpagelayout=OneColumn, pdfnewwindow=true, pdfstartview=XYZ, plainpages=false}

\makeatletter

\providecommand{\tabularnewline}{\\}
\floatstyle{ruled}
\newfloat{algorithm}{tbp}{loa}
\providecommand{\algorithmname}{Algorithm}
\floatname{algorithm}{\protect\algorithmname}

\theoremstyle{plain}
\newtheorem{lem}{\protect\lemmaname}
\ifx\proof\undefined\
  \newenvironment{proof}[1][\proofname]{\par
    \normalfont\topsep6\p@\@plus6\p@\relax
    \trivlist
    \itemindent\parindent
    \item[\hskip\labelsep
          \scshape
      #1]\ignorespaces
  }{%
    \endtrivlist\@endpefalse
  }
  \providecommand{\proofname}{Proof}
\fi
\theoremstyle{definition}
\newtheorem{defn}{\protect\definitionname}
\theoremstyle{plain}
\newtheorem{thm}{\protect\theoremname}
\theoremstyle{definition}
\newtheorem{example}{\protect\examplename}
\theoremstyle{remark}
\newtheorem{rem}{\protect\remarkname}
\theoremstyle{plain}
\newtheorem{prop}{\protect\propositionname}

\usepackage{epstopdf}

\usepackage{tikz}
\usetikzlibrary{arrows}
\usetikzlibrary{shadows}
\usetikzlibrary{positioning}

\def\squarecorner#1{
    \pgf@x=\the\wd\pgfnodeparttextbox%
    \pgfmathsetlength\pgf@xc{\pgfkeysvalueof{/pgf/inner xsep}}%
    \advance\pgf@x by 2\pgf@xc%
    \pgfmathsetlength\pgf@xb{\pgfkeysvalueof{/pgf/minimum width}}%
    \ifdim\pgf@x<\pgf@xb%
        \pgf@x=\pgf@xb%
    \fi%
    \pgf@y=\ht\pgfnodeparttextbox%
    \advance\pgf@y by\dp\pgfnodeparttextbox%
    \pgfmathsetlength\pgf@yc{\pgfkeysvalueof{/pgf/inner ysep}}%
    \advance\pgf@y by 2\pgf@yc%
    \pgfmathsetlength\pgf@yb{\pgfkeysvalueof{/pgf/minimum height}}%
    \ifdim\pgf@y<\pgf@yb%
        \pgf@y=\pgf@yb%
    \fi%
    \ifdim\pgf@x<\pgf@y%
        \pgf@x=\pgf@y%
    \else
        \pgf@y=\pgf@x%
    \fi
    \pgf@x=#1.5\pgf@x%
    \advance\pgf@x by.5\wd\pgfnodeparttextbox%
    \pgfmathsetlength\pgf@xa{\pgfkeysvalueof{/pgf/outer xsep}}%
    \advance\pgf@x by#1\pgf@xa%
    \pgf@y=#1.5\pgf@y%
    \advance\pgf@y by-.5\dp\pgfnodeparttextbox%
    \advance\pgf@y by.5\ht\pgfnodeparttextbox%
    \pgfmathsetlength\pgf@ya{\pgfkeysvalueof{/pgf/outer ysep}}%
    \advance\pgf@y by#1\pgf@ya%
}
\makeatother

\pgfdeclareshape{square}{
    \savedanchor\northeast{\squarecorner{}}
    \savedanchor\southwest{\squarecorner{-}}

    \foreach \x in {east,west} \foreach \y in {north,mid,base,south} {
        \inheritanchor[from=rectangle]{\y\space\x}
    }
    \foreach \x in {east,west,north,mid,base,south,center,text} {
        \inheritanchor[from=rectangle]{\x}
    }
    \inheritanchorborder[from=rectangle]
    \inheritbackgroundpath[from=rectangle]
}

\usetikzlibrary{shapes}
\usetikzlibrary{arrows,decorations.pathmorphing}

\usepackage{hyperref}

\usepackage{pdfsync}

\usepackage{mathtools,amssymb,lipsum}

\usepackage{cuted}
\usepackage{bbm}

\usepackage{algorithm,algpseudocode}
\usepackage{mathdots}  
\usepackage{lipsum}
\usepackage{mathtools}
\usepackage{cuted}

\makeatother

\providecommand{\definitionname}{Definition}
\providecommand{\examplename}{Example}
\providecommand{\lemmaname}{Lemma}
\providecommand{\propositionname}{Proposition}
\providecommand{\remarkname}{Remark}
\providecommand{\theoremname}{Theorem}

\begin{document}
\title{Distributed Neighbor Selection\\
in Multi-agent Networks}
\author{Haibin~Shao,~\IEEEmembership{Member,~IEEE,} Lulu Pan, Mehran~Mesbahi,~\IEEEmembership{Fellow,~IEEE,}\\
Yugeng~Xi,~\IEEEmembership{Senior Member,~IEEE} and Dewei~Li\thanks{This work was supported by the National Science Foundation of China
(Grant No.62103278, 61973214,  61963030). }\thanks{H.\,Shao, L.\,Pan, Y.\,Xi and D. Li are with the Department of
Automation, Shanghai Jiao Tong University, and Key Laboratory of System
Control and Information Processing, Ministry of Education of China,
and Shanghai Engineering Research Center of Intelligent Control and
Management, Shanghai, 200240, China. }\thanks{M.\,Mesbahi is with the William E. Boeing Department of Aeronautics
and Astronautics, University of Washington, Seattle, WA, 98195-2400,
USA. }\thanks{Corresponding author: Lulu Pan (llpan@sjtu.edu.cn).}}
\maketitle
\begin{abstract}
Achieving consensus via nearest neighbor rules is an important prerequisite
for multi-agent networks to accomplish collective tasks. A common
assumption in consensus setup is that each agent interacts with all
its neighbors. This paper examines whether network functionality and
performance can be maintained-and even enhanced-when agents interact
only with a subset of their respective (available) neighbors. As shown
in the paper, the answer to this inquiry is affirmative. In this direction,
we show that by exploring the monotonicity property of the Laplacian
eigenvectors, a neighbor selection rule with guaranteed performance
enhancements, can be realized for consensus-type networks. For distributed
implementation, a quantitative connection between entries of Laplacian
eigenvectors and the \textquotedblleft relative rate of change\textquotedblright{}
in the state between neighboring agents is further established; this
connection facilitates a distributed algorithm for each agent to identify
\textquotedblleft favorable\textquotedblright{} neighbors to interact
with. Multi-agent networks with and without external influence are
examined, as well as extensions to signed networks. This paper underscores
the utility of Laplacian eigenvectors in the context of distributed
neighbor selection, providing novel insights into distributed data-driven
control of multi-agent systems.
\end{abstract}


\section{Introduction}

A multi-agent network is composed of a group of agents, interacting
with their respective nearest neighbors by following local rules;
when such local rules lead to an emerging collective behavior at the
network level is of great interest \cite{mesbahi2010graph,qin2016recent,cao2013overview}.
Achieving consensus via pairwise diffusive interactions between neighboring
agents is a prototypical collective behavior of multi-agent systems
\cite{olfati2004consensus,Jadbabaie2003,lin2005necessary,ren2005consensus},
which also turns out to be a critical prerequisite in disciplines
such as distributed control of networked systems \cite{chung2018survey,song2020network},
distributed estimation over sensor networks \cite{barooah2007estimation},
synchronization in complex networks \cite{dorfler2013synchronization},
large-scale multi-agent machine learning \cite{nedic2020distributed},
and opinion dynamics \cite{proskurnikov2018tutorial}. 

\subsection{Motivation}

The functionality and performance of a multi-agent network are dependent
on the underlying network topology, realized via each agent's interactions
with its nearest neighbors \cite{olfati2004consensus,Jadbabaie2003,ren2005consensus}.
In practice, information exchange using communication channels amongst
agents are often expensive \cite{vasarhelyi2018optimized}. Here,
an important question is whether the network performance can be maintained
and even enhanced if agents interact only with a subset of their available
neighbors, namely, via a neighbor selection. For instance, in leader-follower
multi-robotic networks, it is often assumed that each robot interacts
with all robots within a sensing radius; an important observation
is that the resultant network is not necessary efficient in diffusion
of information from leader robots to the followers \cite{chung2018survey}.
A similar situation arises in distributed optimization and estimation,
where the coordination network is realized via the spatial distribution
of processors or sensors which may not be optimal for the given tasks,
leading to a performance loss \cite{yang2019survey}. 

In fact, the neighbor selection is ubiquitous in both natural and
artificial networks. For instance, it has been reported that in flocks
of starlings, birds interact only with a subset of their nearest neighbors,
rather than with all birds within a sensing radius \cite{ballerini2008interaction}.
An analogous scenario is observed in social networks, where an individual
often determines the subset of their friends to interact with on online
social media; this phenomenon also occurs in real-world social interactions
amongst people \cite{allen2017evolutionary}. Neighbor selection schemes
are also employed in peer-to-peer networks, such as BitTorrent, to
save traffic overhead \cite{bindal2006improving}. Along the same
lines, adaptive neighbor selection has been proposed to enhance the
quality of predicted ratings in recommender systems \cite{ahmadian2018social}.
The $k$NN imputation methods are designed to select $k$ nearest
neighbors to deal with missing data in datasets \cite{zhang2012nearest}.
Notably, neighbor selection can also be viewed as an attention mechanism
(each agent pays more attention to specific agents), which is ubiquitously
employed in recently developed learning  algorithms \cite{vaswani2017attention}. 

For multi-agent consensus problems, network topology plays a crucial
role in both reaching consensus and the corresponding  convergence
rate \cite{olshevsky2009convergence,clark2018maximizing,kim2006maximizing}.
A common assumption in this line of work is that each agent interacts
with all its neighbors \cite{mesbahi2010graph,qin2016recent,cao2013overview};
however, there may exist excessive interactions that degrade the performance
of the multi-agent network. A natural question thereby is whether
the importance of agents' neighbors (with respect to the desired performance)
can be inferred from local measurements. This (data-driven) distributed
neighbor selection problem is the focus of the present work. Our work
is also inspired by the observation that information flow between
a pair of neighboring agents does not need to be bidirectional, especially
when two neighboring agents are not hierarchical equivalent \cite{dedeo2021equality}.
For instance, a rooted tree exhibits a typical hierarchical structure
for the efficient spreading of information from the root to other
nodes- a bidirectional information exchange on the other hand may
lessen the efficiency of the convergence process. In particular, we
provide a theoretical framework to reason about the distributed neighbor
selection problem as well as a guarantee of its performance. As we
will show, specific Laplacian eigenvectors facilitate a systematic
treatment for designing and analyzing a novel distributed neighbor
selection algorithm for consensus-type networks.

\subsection{Contribution}

The contribution of this paper is threefold. First, we will show that
neighbor selection (effectively removing a specific subset of edges
from the network) can be effective for improving the network performance.
In this direction, one of our contributions involves using entries
of the Laplacian eigenvector as a criterion for neighbor selection;
subsequently, we show how the obtained reduced network maintains,
and even enhances the functionality of the original network in terms
of network reachability. Secondly, inspired by the observation that
bidirectional interactions amongst neighboring agents can hinder the
efficiency of information propagation throughout a consensus-type
network, we provide theoretical guarantees on the performance enhancement
of the reduced network in terms of convergence rate. Finally, as the
Laplacian eigenvectors are global network variables, we establish
a quantitative connection between the entries of the Laplacian eigenvector
and the relative rate of change in state between a pair of neighboring
agents, a quantity that we have referred to as the network relative
tempo. An important observation is that relative tempo is computable
from local measurements. In this direction, we show how the relative
tempo can be employed for the distributed online neighbor selection
process. 

The contributions of this work have several immediate consequences.
First, linking local interactions and global collective behaviors
is a central topic in complex systems, this work essentially initiates
a novel local neighbor selection rule that leads to a reduced network
with guaranteed network reachability and enhanced convergence rate
at the network level \cite{vicsek2012collective,beni2020swarm}. Second,
as consensus-type networks play a central role in distributed algorithms,
our work has immediate consequences for distributed control, estimation,
optimization, and learning algorithm design \cite{kia2019tutorial,nedic2020distributed}.
Third, the quantitative connection between entries of Laplacian eigenvectors
and relative tempo provides novel insights into distributed online
control of multi-agent networks \cite{brunton2019data}. Lastly, but
certainly not least, as opposed to Laplacian eigenvalues (e.g., algebraic
connectivity \cite{fiedler1973algebraic}) that have been extensively
examined in graph theory literature and for consensus problems \cite{olfati2004consensus,pirani2016smallest,clark2018maximizing},
this paper underscores the utility of the Laplacian eigenvectors (namely,
Fiedler vector and its variant) by unveiling the network reachability
that they encode (e.g., further extending the celebrated results of
Fiedler in \cite{Fiedler1975}), a novel application of Fiedler vector
in addition to spectral clustering \cite{Fiedler1975,von2007tutorial,merris1998laplacian}. 

\subsection{Organization}

The remainder of this paper is organized as follows. We introduce
preliminaries covering notation, graph theory, and network dynamics
in $\mathsection$\ref{sec:preliminaries}. A motivational example
is then provided and discussed in $\mathsection$\ref{sec:A-Motivating-Example}.
The main results for semi-autonomous networks in terms of analysis
of reachability of reduced networks after neighbor selection process,
as well as the corresponding convergence rates, are provided in $\mathsection$\ref{sec:SAN};
this is then followed by parallel results for fully autonomous networks
in $\mathsection$\ref{sec:FAN}. Extensions of main results to signed
networks are discussed in $\mathsection$\ref{sec:Extension-to-Signed},
followed by concluding remarks in $\mathsection$\ref{sec:Conclusion-Remarks}. 

\section{Preliminaries \label{sec:preliminaries}}

First a quick note on the notation. Let $\mathbb{R}$ and $\mathbb{Z}_{+}$
denote the set of real numbers and positive integers, respectively.
Denote the set $\left\{ 1,2,\ldots,n\right\} $ as $\underline{n}$,
where $n\in\mathbb{Z}_{+}$; $\mathds{1}_{n}$ and $\boldsymbol{0}_{n\times m}$
denote $n\times1$ vector and $n\times m$ matrix of all ones and
all zeros, respectively. Let $I_{d}$ denote the $d\times d$ identity
matrix and $\boldsymbol{e}_{j}$ denote the $j$th column of $I_{d}$
where $j\in\underline{d}$. The $i$th smallest eigenvalue and the
corresponding normalized eigenvector of a symmetric matrix $M\in\mathbb{R}^{n\times n}$
is signified by $\lambda_{i}(M)$ and $\boldsymbol{v}_{i}(M)$, respectively.
The entry located at the $i$th row and $j$th column in a matrix
$M\in\mathbb{R}^{n\times n}$ is denoted by $[M]_{ij}$ and the $i$th
entry of a vector $\boldsymbol{x}$ by $[\boldsymbol{x}]_{i}$. Let
$\boldsymbol{x}_{ij}$ denote $\frac{[\boldsymbol{x}]_{i}}{[\boldsymbol{x}]_{j}}$
for a vector $\boldsymbol{x}\in\mathbb{R}^{n}$. The Euclidean norm
of a vector $\boldsymbol{x}\in\mathbb{R}^{n}$ is designated by $\|\boldsymbol{x}\|=(\boldsymbol{x}^{\top}\boldsymbol{x})^{\frac{1}{2}}$
. A vector $\boldsymbol{x}\in\mathbb{R}^{n}$ is positive if $[\boldsymbol{x}]_{i}>0$
for all $i\in\underline{n}$. The spectral radius of a matrix $M$
is denoted by $\text{\ensuremath{\rho}}(M)$. 

Next, we provide a few graph-theoretic constructs that will be subsequently
used in the paper. Let $\mathcal{G}=(\mathcal{V},\mathcal{E},W)$
denote a graph with the node set $\mathcal{V}=\left\{ 1,2,\ldots,n\right\} $
and edge set $\mathcal{E\subset V\times V}$. The adjacency matrix
$W=(w_{ij})\in\mathbb{R}^{n\times n}$ is such that the edge weight
between agents $i$ and $j$ satisfies $w_{ij}\not=0$ if and only
if $(i,j)\in\mathcal{E}$ and $w_{ij}=0$ otherwise. A graph $\mathcal{G}$
is undirected if $(i,j)\in\mathcal{E}$ if and only if $(j,i)\in\mathcal{E}$,
otherwise $\mathcal{G}$ is directed. A graph $\mathcal{G}$ is a
signed graph if there exists an edge $(i,j)\in\mathcal{E}$ such that
$w_{ij}<0$, otherwise $\mathcal{G}$ is unsigned.  Let $\mathcal{N}_{i}=\left\{ j\in\mathcal{V}|(i,j)\in\mathcal{E}\right\} $
denote the neighbor set of an agent $i\in\mathcal{V}$. A path from
$i_{p}\text{\ensuremath{\in\mathcal{V}}}$ to $i_{1}\text{\ensuremath{\in\mathcal{V}}}$
in $\mathcal{G}$ is a concatenation of edges $(i_{1},i_{2}),(i_{2},i_{3}),\cdots,(i_{p-1},i_{p})$,
where all nodes $i_{1},i_{2},\ldots,i_{p}$ are distinct; a node $i\in\mathcal{V}$
is reachable from a node $j\in\mathcal{V}$ if there exists a path
from $j$ to $i$ in $\mathcal{G}$. An undirected graph is connected
if each pair of nodes are reachable from each other.  Let $S_{n}$
denote a star graph with $n\in\mathbb{Z}_{+}$ nodes. A subgraph $\mathcal{\tilde{G}}=(\tilde{\mathcal{V}},\tilde{\mathcal{E}})$
of a graph $\mathcal{G}=(\mathcal{V},\mathcal{E})$ is a graph such
that $\tilde{\mathcal{V}}\subset\mathcal{V}$ and $\tilde{\mathcal{E}}\subset\mathcal{E}$.
The subgraph obtained by removing a node set $\mathcal{V}^{\prime}\subset\mathcal{V}$
and all incident edges from a graph $\mathcal{G}=(\mathcal{V},\mathcal{E})$
is denoted by $\mathcal{G}-\mathcal{V}^{\prime}.$ Let $\mathcal{S}\subset\mathcal{V}$
be any subset of nodes in $\mathcal{G}=(\mathcal{V},\mathcal{E})$.
Then the induced subgraph $\mathcal{G}(\mathcal{S})$ is the graph
whose node set is $\mathcal{S}$ and whose edge set consists of all
of the edges incident to nodes in $\mathcal{S}$.
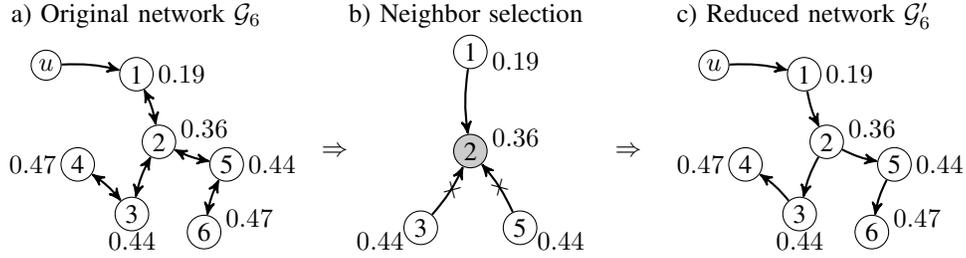
\begin{figure*}[t]
\begin{centering}
\begin{tikzpicture}[scale=0.6, >=stealth',   pos=.8,  photon/.style={decorate,decoration={snake,post length=1mm}} ] 	
	\node (n1) at (-0.5,1.5) [circle,inner sep= 1.5pt,draw] {1}; 	
	\node (n2) at (0,0) [circle,inner sep= 1.5pt,draw] {2};     
	\node (n3) at (-0.6,-1.6) [circle,inner sep= 1.5pt,draw] {3};     
	\node (n4) at (-1.8,-.5) [circle,inner sep= 1.5pt,draw] {4}; 	
	\node (n5) at (1.5,-0.5) [circle,inner sep= 1.5pt,draw] {5}; 	
	\node (n6) at (1,-2) [circle,inner sep= 1.5pt,draw] {6};
	
    \node (u) at (-2.5,1.7) [circle,inner sep= 1.5pt,draw] {$u$};
    \node (a) at (-0.5,2.8) {a) Original network $\mathcal{G}_6$};
    \node (nn1) at (0.5,1.5) {$0.19$};    
	\node (nn2) at (1,0.3) {$0.36$};    
	\node (nn3) at (-0.6,-2.2) {$0.44$};    
	\node (nn4) at (-2.8,-.5) {$0.47$};    
	\node (nn5) at (2.5,-0.5) {$0.44$};    
	\node (nn6) at (2,-1.7) {$0.47$};

	\path[] 	(u) [->,thick, bend left=8] edge node[below] {} (n1);
	\path[] 	(n2) [<->,thick, bend right=8] edge node[below] {} (n1);
	\path[] 	(n2) [<->,thick, bend right=8] edge node[below] {} (n3);
	\path[] 	(n3) [<->,thick, bend right=8] edge node[below] {} (n4);
	\path[] 	(n2) [<->,thick, bend right=8] edge node[below] {} (n5);
	\path[] 	(n5) [<->,thick, bend right=8] edge node[below] {} (n6);
\end{tikzpicture}\,\,\begin{tikzpicture}[scale=0.6, >=stealth',   pos=.8,  photon/.style={decorate,decoration={snake,post length=1mm}} ]
	\node (n1) at (0,2.2) [circle,inner sep= 1.5pt,draw] {1};
	\node (n2) at (0,0) [circle,inner sep= 1.5pt,fill=black!20,draw] {2};
    \node (n3) at (-1.1,-1.7) [circle,inner sep= 1.5pt,draw] {3};
	\node (n5) at (1.1,-1.7) [circle,inner sep= 1.5pt,draw] {5};

   \node (nn1) at (1,2) {$0.19$};
   \node (nn2) at (1,0.3) {$0.36$};
   \node (nn3) at (-2,-2) {$0.44$};
   \node (nn5) at (2,-2) {$0.44$};

    \node (a) at (-0.1,3) {b) Neighbor selection};

    \node (r1) at (-3,0.0) {$\Rightarrow$};
    \node (r2) at (3.5,0.0) {$\Rightarrow$};

	\node (x2) [rotate=60] at (-0.4,-0.85) {$\times$};
	\node (x3) [rotate=30] at  (0.65,-0.8) {$\times$};

	\path[]
	(n1) [->,thick, bend right=8] edge node[below] {} (n2);

	\path[]
	(n3) [->,thick, bend right=8] edge node[below] {} (n2);

	\path[]
	(n5) [->,thick, bend right=8] edge node[below] {} (n2);

\end{tikzpicture}\,\,\,\,\begin{tikzpicture}[scale=0.6, >=stealth',   pos=.8,  photon/.style={decorate,decoration={snake,post length=1mm}} ]
	\node (n1) at (-0.5,1.5) [circle,inner sep= 1.5pt,draw] {1};
	\node (n2) at (0,0) [circle,inner sep= 1.5pt,draw] {2};
    \node (n3) at (-0.6,-1.6) [circle,inner sep= 1.5pt,draw] {3};
    \node (n4) at (-1.8,-.5) [circle,inner sep= 1.5pt,draw] {4};
	\node (n5) at (1.5,-0.5) [circle,inner sep= 1.5pt,draw] {5};
	\node (n6) at (1,-2) [circle,inner sep= 1.5pt,draw] {6};

    \node (u) at (-2.5,1.7) [circle,inner sep= 1.5pt,draw] {$u$};

    \node (a) at (-0.5,2.8) {c) Reduced network $\mathcal{G}_6^{\prime}$};


   \node (nn1) at (0.5,1.5) {$0.19$};
   \node (nn2) at (1,0.3) {$0.36$};
   \node (nn3) at (-0.6,-2.2) {$0.44$};
   \node (nn4) at (-2.8,-.5) {$0.47$};
   \node (nn5) at (2.5,-0.5) {$0.44$};
   \node (nn6) at (2,-1.7) {$0.47$};

	\path[]
	(u) [->,thick, bend left=8] edge node[below] {} (n1);

	\path[]

	(n1) [->,thick, bend right=8] edge node[below] {} (n2);

	\path[]
	(n2) [->,thick, bend right=8] edge node[below] {} (n3);

	\path[]
	(n3) [->,thick, bend right=8] edge node[below] {} (n4);

	\path[]
	(n2) [->,thick, bend right=8] edge node[below] {} (n5);

	\path[]
	(n5) [->,thick, bend right=8] edge node[below] {} (n6);

\end{tikzpicture}
\par\end{centering}
\caption{The orientation of each edge indicates the direction of information
flow or influence. For instance, the edge from $u$ to agent $1$
implies that the information of $u$ can be transmitted from $u$
to agent $1$ and subsequently  agent $1$ can be influenced by $u$.
The bidirectional edges between neighboring agents are identified
by the line with double arrows for simplicity. Entries of $\boldsymbol{v}_{1}(L_{B}(\mathcal{G}_{6}))$
associated to each agent in $\mathcal{G}_{6}$ are indicated by numbers
close to each agent.}

\label{fig:motivation-example-6-node}
\end{figure*}

Lastly, we provide a brief synopsis of multi-agent networks\footnote{We will use \textquotedblleft graphs\textquotedblright{} and \textquotedblleft networks\textquotedblright{}
interchangeably in this paper.}. In a multi-agent network $\mathcal{G}=(\mathcal{V},\mathcal{E},W)$,
each agent $i\in\mathcal{V}$ has the state $\boldsymbol{x}_{i}(t)\in\mathbb{R}^{d}$
(or $\boldsymbol{x}_{i}\in\mathbb{R}^{d}$) at time $t$. In the sequel
we will consider two distinct categories of diffusively coupled networks,
namely, fully-autonomous and semi-autonomous.

In a fully autonomous network (FAN), $n\in\mathbb{Z}_{+}$ agents
evolve their respective states through interactions characterized
by an unsigned graph $\mathcal{G}=(\mathcal{V},\mathcal{E},W)$. In
particular, each agent updates its state by adopting the diffusive
interaction protocol (extensively examined in distributed algorithms
and synchronization problems \cite{chung2018survey,song2020network,barooah2007estimation,dorfler2013synchronization,nedic2020distributed}),
\begin{equation}
\dot{\boldsymbol{x}}_{i}(t)=-\sum_{i=1}^{n}w_{ij}\left(\boldsymbol{x}_{i}(t)-\boldsymbol{x}_{j}(t)\right),\text{ }i\in\mathcal{V}.\label{eq:consensus-protocol}
\end{equation}
In relation with the protocol (\ref{eq:consensus-protocol}), we say
the state of agent $i$ is influenced by its neighbors $j\in\mathcal{N}_{i}$
or, equivalently, agent $i$ follows its neighbors $j\in\mathcal{N}_{i}$.
Denote the graph Laplacian of $\mathcal{G}$ as $L(\mathcal{G})=(l_{ij})\in\mathbb{R}^{n\times n}$
where $l_{ii}=\sum{}_{j=1}^{n}w_{ij}$ for $i\in\mathcal{V}$ and
$l_{ij}=-w_{ij}$ for $i\ne j$. The collective behavior of a FAN
can be characterized as,
\begin{equation}
\dot{\boldsymbol{x}}=-(L(\mathcal{G})\otimes I_{d})\boldsymbol{x},\label{eq:consensus-overall}
\end{equation}
where $\boldsymbol{x}=(\boldsymbol{x}_{1}^{\top}(t),\dots,\boldsymbol{x}_{n}^{\top}(t))^{\top}\in\mathbb{R}^{nd}$.

In semi-autonomous networks (SANs), a subset of agents (referred to
as leaders or informed agents) are selected to receive external control
signals so as to steer the entire network towards a desired state.
In this direction, consider a SAN consisting of $n\in\mathbb{Z}_{+}$
agents whose interaction structure is characterized by an unsigned
graph $\mathcal{G}=(\mathcal{V},\mathcal{E},W)$. In a SAN, the leaders,
denoted by $\mathcal{V}_{\text{leader}}\subset\mathcal{V}$, can be
directly influenced by the external input signals and the remaining
agents are referred to as followers, denoted by $\mathcal{V}_{\text{follower}}=\mathcal{V}\setminus\mathcal{V}_{\text{leader}}$.
In this paper, the set of external inputs is denoted by $\mathcal{U}=\left\{ \boldsymbol{u}_{1},\dots,\boldsymbol{u}_{m}\right\} $,
where $\boldsymbol{u}_{l}\in\mathbb{R}^{d}$, $l\in\underline{m}$
and $m\in\mathbb{Z}_{+}$. Then $\mathcal{U}$ is homogeneous if $\boldsymbol{u}_{i}=\boldsymbol{u}_{j}$
for all $i\ne j\in\underline{m}$ and heterogeneous if otherwise.
In this setup, it is assumed that each leader is at most influenced
by one external input. In a SAN, the interaction protocol for each
agent $i\in\mathcal{V}$ admits the form,
\begin{align}
\dot{\boldsymbol{x}}_{i}(t) & =-\sum_{i=1}^{n}w_{ij}(\boldsymbol{x}_{i}(t)-\boldsymbol{x}_{j}(t))-\sum_{l=1}^{m}b_{il}(\boldsymbol{x}_{i}(t)-\boldsymbol{u}_{l}),\label{eq:LF-unsigned-protocol}
\end{align}
where $b_{il}=1$ if and only if $i\in\mathcal{V}_{\text{leader}}$
and $b_{il}=0$ otherwise\footnote{The SAN (\ref{eq:LF-unsigned-protocol}) is also known as Taylor\textquoteright s
model in social network analysis \cite{proskurnikov2018tutorial}.}. Subsequently, the collective behavior of SAN (\ref{eq:LF-unsigned-protocol})
can be characterized as, 
\begin{equation}
\dot{\boldsymbol{x}}=-(L_{B}(\mathcal{G})\otimes I_{d})\boldsymbol{x}+(B\otimes I_{d})\boldsymbol{u},\label{eq:unsigned-LF-overall}
\end{equation}
where $\boldsymbol{x}=(\boldsymbol{x}_{1}^{\top}(t),\dots,\boldsymbol{x}_{n}^{\top}(t))^{\top}\in\mathbb{R}^{nd}$,
$B=(b_{il})\in\mathbb{R}^{n\times m}$, $\boldsymbol{u}=(\boldsymbol{u}_{1}^{\top},\dots,\boldsymbol{u}_{m}^{\top})^{\top}\in\mathbb{R}^{md}$
and 
\begin{equation}
L_{B}(\mathcal{G})=L(\mathcal{G})+\text{{\bf diag}}(B\mathds{1}_{m}),\label{eq:LB-matrix}
\end{equation}
which is referred to as perturbed Laplacian since $L_{B}(\mathcal{G})$
(or $L_{B}$ for brevity) is obtained from a perturbation on the Laplacian
matrix by a diagonal matrix $\text{{\bf diag}}(B\mathds{1}_{m})$
\cite{pirani2016smallest,xia2017analysis,clark2018maximizing,cao2012distributed,Airlie2013TAC,dorfler2018electrical}.
The FAN (\ref{eq:consensus-overall}) or SAN (\ref{eq:unsigned-LF-overall})
are said to achieve consensus if ${\displaystyle \lim_{t\rightarrow\infty}}\|\boldsymbol{x}_{i}(t)-\boldsymbol{x}_{j}(t)\|=0$
for all $i,j\in\mathcal{V}$ and some norm on $\mathbb{R}^{d}$ \cite{olfati2004consensus,cao2012distributed}.
We assume throughout this paper that the underlying networks of FAN
(\ref{eq:consensus-overall}) and SAN (\ref{eq:unsigned-LF-overall})
are all undirected and connected before implementing neighbor selection.

\section{A Motivational Scenario \label{sec:A-Motivating-Example}}

We provide an example to motivate this work. Consider a SAN on a connected
unsigned network $\mathcal{G}_{6}$ in Figure \ref{fig:motivation-example-6-node}a
(referred to as original network), where agent $1$ is a leader with
an external input $u=0.9$.  We know that reachability (existence
of a directed path) from the external input to each agent is a prerequisite
for the agents to track this external input $u$. This observation
motivates us to inquire whether the remaining edges, apart from those
that can guarantee the leader-follower reachability of the network,
are necessary for reaching a consensus. For instance, the network
$\mathcal{G}_{6}^{\prime}$ (Figure \ref{fig:motivation-example-6-node}c)
is the minimal subgraph of $\mathcal{G}_{6}$ (in terms of the number
of edges) that can guarantee the reachability from external input
$u$ to all the agents, the corresponding  convergence performance
is significantly enhanced compared with that of the original network
$\mathcal{G}_{6}$ (see Figure \ref{fig:motivation-trajectory-6-node}).
\begin{figure}[tbh]
\begin{centering}
\includegraphics[width=8cm]{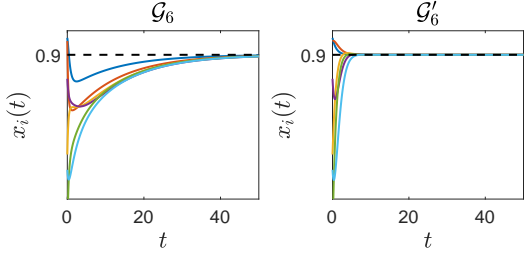}
\par\end{centering}
\caption{State trajectories of agents in a SAN (\ref{eq:unsigned-LF-overall})
evolving on networks $\mathcal{G}_{6}$ (left) and $\mathcal{G}_{6}^{\prime}$
(right) in Figure \ref{fig:motivation-example-6-node}, respectively.}
\label{fig:motivation-trajectory-6-node}
\end{figure}
 Apparently,  the reduced network $\mathcal{G}_{6}^{\prime}$ can
be constructed from $\mathcal{G}_{6}$ by eliminating one of the bidirectional
edges between neighboring agents in $\mathcal{G}_{6}$ (see Figure
\ref{fig:motivation-example-6-node}b), how the local accessible information
of each agent can be employed to guide this neighbor selection process
is challenging. 

Recall that the eigenvector associated with the smallest non-zero
eigenvalue of perturbed Laplacian can be chosen to be positive. An
important observation is that the entries of this eigenvector, along
the directed paths from leader agents to all follower agents, are
monotonically increasing (subsequently, we will show that this is
not accidental). We will subsequently see that this monotonicity property
plays an important role in the distributed neighbor selection process
(see Figure \ref{fig:motivation-example-6-node}b). In the sequel
we first examine this observation analytically for SANs, followed
by its implications for FANs.

\section{Semi-Autonomous Networks \label{sec:SAN}}

In this section, a neighbor selection algorithm, based on the monotonicity
of the eigenvector entries associated with the perturbed Laplacian,
is proposed. Subsequently, the convergence rate of the multi-agent
system on the reduced network, post neighbor selection process, will
be examined. Furthermore, the distributed implementation of the neighbor
selection process is discussed.

\subsection{Reachability Analysis}

We shall first examine the reachability property of SANs, as encoded
in a Laplacian eigenvector of the underlying network. The eigen-pair
$(\lambda_{\text{1}}(L_{B}),\boldsymbol{v}_{\text{1}}(L_{B}))$ associated
with the perturbed Laplacian $L_{B}$ of a SAN, with input matrix
$B$, turns out to be an important algebraic construct revealing graph-theoretic
properties of SANs. As such, we shall unveil the network reachability,
encoded in the eigenvector $\boldsymbol{v}_{\text{1}}(L_{B})$, providing
useful insights for designing neighbor selection algorithm for SANs. 

First, we provide some preliminary properties of $(\lambda_{\text{1}}(L_{B}),\boldsymbol{v}_{\text{1}}(L_{B}))$. 
\begin{lem}
\label{lem:uniqueness-unsigned} Let $\lambda_{1}(L_{B})$ and $\boldsymbol{v}_{1}(L_{B})$
denote the smallest eigenvalue and the corresponding normalized  eigenvector
of $L_{B}$ in (\ref{eq:LB-matrix}), respectively. Then, $\lambda_{1}(L_{B})>0$
is a simple eigenvalue of $L_{B}$ and $\boldsymbol{v}_{1}(L_{B})$
can be chosen to be (component-wise) positive.
\end{lem}
\begin{proof}
Refer to the Appendix.
\end{proof}
For a SAN $\mathcal{G}$ with the input matrix $B$, we proceed to
construct a reduced subgraph of $\mathcal{G}$ by eliminating a subset
of edges between an agent and its neighboring agents, using information
encoded in $\boldsymbol{v}_{\text{1}}(L_{B})$, namely, realizing
neighbor selection. We shall refer to this class of reduced subgraphs
as following the slower neighbor (FSN) networks of $\mathcal{G}$,
since it is implied that each agent follows (or chooses to be influenced
by) those neighbors whose rate of change in states are relatively
slower; this statement will be made rigorous in $\mathsection$\ref{subsec:NS-SAN}.
\begin{defn}[\textbf{FSN network of SANs}]
\label{def:fsn-network-SAN} Let $\mathcal{G}=(\mathcal{V},\mathcal{E},W)$
be an unsigned SAN with the input matrix $B$. The FSN network of
$\mathcal{G}$, denoted by $\bar{\mathcal{G}}=(\mathcal{\bar{V}},\bar{\mathcal{E}},\bar{W})$,
is a subgraph of $\mathcal{G}$ such that $\mathcal{\bar{V}}=\mathcal{V}$,
$\mathcal{\bar{E}}\subseteq\mathcal{E}$ and $\bar{W}=(\bar{w}_{ij})\in\mathbb{R}^{n\times n}$,
where $\bar{w}_{ij}=w_{ij}$ if $\boldsymbol{v}_{\text{1}}(L_{B})_{ij}>1$
and $\bar{w}_{ij}=0$ when $\boldsymbol{v}_{\text{1}}(L_{B})_{ij}\leq1$.
\end{defn}
According to  Definition \ref{def:fsn-network-SAN}, the FSN network
of a SAN is determined by the perturbed Laplacian, specifically by
the corresponding eigenvector $\boldsymbol{v}_{\text{1}}(L_{B})$.
Note that the construction of the FSN network is essentially achieved
by comparing $\boldsymbol{v}_{\text{1}}(L_{B})_{ij}$ and $1$; hence
this process can be regarded as a ``rule'' of neighbor selection.
We shall now proceed to reveal the leader-to-follower reachability
(LF-reachability) of FSN networks. 
\begin{thm}
\label{thm:reachability-FSN-SAN} Let $\bar{\mathcal{G}}=(\bar{\mathcal{V}},\mathcal{\bar{E}},\bar{W})$
be the FSN network of the SAN (\ref{eq:unsigned-LF-overall}) on the
unsigned connected network $\mathcal{G}=(\mathcal{V},\mathcal{E},W)$.
Then for each agent $i\in\bar{\mathcal{V}}$, there exists $l\in\underline{m}$
such that $i$ is reachable from $\boldsymbol{u}_{l}$ in $\bar{\mathcal{G}}$. 
\end{thm}
\begin{proof}
Let  $\mathcal{V}_{\text{\text{leader}}}$ and $\mathcal{V}_{\text{\text{follower}}}$
be the leader and the follower sets of the SAN (\ref{eq:unsigned-LF-overall}),
respectively. According to Lemma \ref{lem:steady-state-unsigned-semi}
in Appendix, it is sufficient to show that for an arbitrary $i\in\mathcal{V}_{\text{follower}}$,
there exists a leader agent $l\in\mathcal{V}_{\text{leader}}$ such
that $i$ is reachable from $l$. 

By contradiction, assume that there exists a subset of agents $\left\{ i_{1},i_{2},\cdots,i_{s}\right\} \subset\mathcal{V}_{\text{follower}}$
in the FSN network $\bar{\mathcal{G}}$ such that $i_{k}$ is not
reachable from any $l\in\mathcal{V}_{\text{leader}}$, where $k\in\underline{s}$
and $s\in\mathbb{Z}_{+}$. Let $\lambda_{1}$ be the smallest eigenvalue
of the perturbed Laplacian matrix $L_{B}(\mathcal{G})$ with the corresponding
eigenvector $\boldsymbol{v}_{1}$. According to Lemma \ref{lem:uniqueness-unsigned},
one has $\lambda_{1}>0$ and the corresponding eigenvector $\boldsymbol{v}_{1}$
is positive. Now consider the following two cases:

Case 1: There exists an isolated agent $i^{\prime}\in\left\{ i_{1},i_{2},\cdots,i_{s}\right\} $
such that agent $i^{\prime}$ is not reachable from any leader agent
in the FSN network $\bar{\mathcal{G}}$. Then, according to Definition
\ref{def:fsn-network-SAN}, one has,
\begin{equation}
[\boldsymbol{v}_{1}]_{i^{\prime}}\le[\boldsymbol{v}_{1}]_{j},\label{eq:eigenvector inequality}
\end{equation}
for all $j\in\mathcal{N}_{i^{\prime}}$. Examining the $i^{\prime}$th
row in eigen-equation $L_{B}(\mathcal{G})\boldsymbol{v}_{1}=\lambda_{1}\boldsymbol{v}_{1}$
yields,
\begin{equation}
\left({\displaystyle \sum_{j\in\mathcal{N}_{i^{\prime}}}}w_{i^{\prime}j}\right)[\boldsymbol{v}_{1}]_{i^{\prime}}-{\displaystyle \sum_{j\in\mathcal{N}_{i^{\prime}}}}w_{i^{\prime}j}[\boldsymbol{v}_{1}]_{j}=\lambda_{1}[\boldsymbol{v}_{1}]_{i^{\prime}}.\label{eq:eigenvector equality}
\end{equation}
Combining (\ref{eq:eigenvector inequality}) and (\ref{eq:eigenvector equality}),
now yields the following inequality,
\begin{equation}
\left({\displaystyle \sum_{j\in\mathcal{N}_{i^{\prime}}}}w_{i^{\prime}j}\right)[\boldsymbol{v}_{1}]_{i^{\prime}}-{\displaystyle \sum_{j\in\mathcal{N}_{i^{\prime}}}}w_{i^{\prime}j}[\boldsymbol{v}_{1}]_{i^{\prime}}\geq\lambda_{1}[\boldsymbol{v}_{1}]_{i^{\prime}}.
\end{equation}
By eliminating $[\boldsymbol{v}_{1}]_{i^{\prime}}>0$ from both sides
of the above inequality, it follows that $\lambda_{1}\leq0$, establishing
a contradiction.

Case 2: There exists a weak connected component $\bar{\mathcal{G}}(\left\{ i_{1},i_{2},\cdots,i_{s_{0}}\right\} )$
in $\left\{ i_{1},i_{2},\cdots,i_{s}\right\} $, such that any agent
in this weak connected component is not reachable from any leader
agent, where $s_{0}\in\mathbb{Z}_{+}$ and $s_{0}\le s$. Let 
\begin{equation}
[\boldsymbol{v}_{1}]_{i^{\prime}}=\underset{k\in\left\{ i_{1},i_{2},\cdots,i_{s_{0}}\right\} }{\min}\left\{ [\boldsymbol{v}_{1}]_{k}\right\} .
\end{equation}
Then, one has $[\boldsymbol{v}_{1}]_{j}\ge[\boldsymbol{v}_{1}]_{i^{\prime}}$
for all $j\in\mathcal{N}_{i^{\prime}}$. Again, one can conclude the
contradiction $\lambda_{1}\leq0$ by applying a similar procedure
as in Case 1. 
\end{proof}
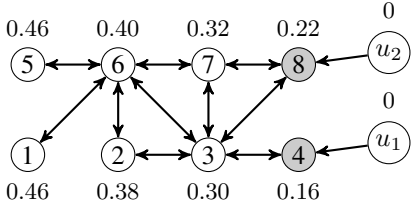
\begin{figure}[tbh]
\begin{centering}
\begin{tikzpicture}[scale=0.8, >=stealth',   pos=.8,  photon/.style={decorate,decoration={snake,post length=1mm}} ]
	\node (n1) at (0,0) [circle,inner sep= 1.5pt,draw] {1};
	\node (n2) at (1.5,0) [circle,inner sep= 1.5pt,draw] {2};
    \node (n3) at (3,0) [circle,inner sep= 1.5pt,draw] {3};
    \node (n4) at (4.5,0) [circle,inner sep= 1.5pt,fill=black!20,draw] {4};
	\node (n5) at (0,1.5) [circle,inner sep= 1.5pt,draw] {5};
	\node (n6) at (1.5,1.5) [circle,inner sep= 1.5pt,draw] {6};
    \node (n7) at (3,1.5) [circle,inner sep= 1.5pt,draw] {7};
    \node (n8) at (4.5,1.5) [circle,inner sep= 1.5pt,fill=black!20,draw] {8};

	\node (TV1) at (0,-0.6)   {{\small $0.46$}};
	\node (TV2) at (1.5,-0.6)   {{\small $0.38$}};
	\node (TV3) at (3,-0.6)   {{\small $0.30$}};
	\node (TV4) at (4.5,-0.6)   {{\small $0.16$}};
	\node (TV5) at (0,2.1)   {{\small $0.46$}};
	\node (TV6) at (1.5,2.1)   {{\small $0.40$}};
	\node (TV7) at (3,2.1)   {{\small $0.32$}};
	\node (TV8) at (4.5,2.1)   {{\small $0.22$}};

    \node (u1) at (6,0.2) [circle,inner sep= 1.5pt,draw] {$u_1$};
    \node (u2) at (6,1.7) [circle,inner sep= 1.5pt,draw] {$u_2$};

	\node (u11) at (6,2.4)   {{\small $0$}};
	\node (u22) at (6,0.9)   {{\small $0$}};

	\path[]
	(u1) [->,thick] edge node[below] {} (n4)	
    (u2) [->,thick] edge node[below] {} (n8);

	\path[]	(n2) [<->,thick] edge node[below] {} (n3);
	\path[]	(n6) [<->,thick] edge node[below] {} (n5);
	\path[] 	(n6) [<->,thick] edge node[below] {} (n1);
	\path[] 	(n2) [<->,thick] edge node[below] {} (n6); 
	\path[] 	(n7) [<->,thick] edge node[below] {} (n6); 
	\path[] 	(n3) [<->,thick] edge node[below] {} (n6);
	\path[] 	(n3) [<->,thick] edge node[below] {} (n7);
	\path[] 	(n8) [<->,thick] edge node[below] {} (n7);
	\path[] 	(n8) [<->,thick] edge node[below] {} (n3);
	\path[] 	(n4) [<->,thick] edge node[below] {} (n3);
\end{tikzpicture}
\par\end{centering}
\caption{An eight-node SAN $\mathcal{G}_{8}$. The entries of $\boldsymbol{v}_{1}(L_{B})$
corresponding to each agent are shown close to each node (with a two
decimal point accuracy).}

\label{fig:8-node-network}
\end{figure}

It turns out that the entries of the eigenvector\textbf{ $\boldsymbol{v}_{1}(L_{B})$}
are influenced by the selection of leader agents. We provide an example
to demonstrate the utility of Theorem \ref{thm:reachability-FSN-SAN}.
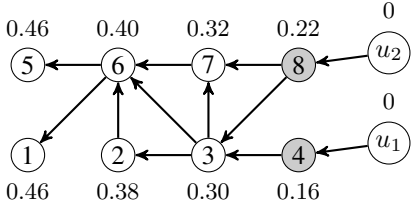
\begin{figure}[tbh]
\centering{}\begin{tikzpicture}[scale=0.8, >=stealth',   pos=.8,  photon/.style={decorate,decoration={snake,post length=1mm}} ]
	\node (n1) at (0,0) [circle,inner sep= 1.5pt,draw] {1};
	\node (n2) at (1.5,0) [circle,inner sep= 1.5pt,draw] {2};
    \node (n3) at (3,0) [circle,inner sep= 1.5pt,draw] {3};
    \node (n4) at (4.5,0) [circle,inner sep= 1.5pt,fill=black!20,draw] {4};
	\node (n5) at (0,1.5) [circle,inner sep= 1.5pt,draw] {5};
	\node (n6) at (1.5,1.5) [circle,inner sep= 1.5pt,draw] {6};
    \node (n7) at (3,1.5) [circle,inner sep= 1.5pt,draw] {7};
    \node (n8) at (4.5,1.5) [circle,inner sep= 1.5pt,fill=black!20,draw] {8};

	\node (TV1) at (0,-0.6)   {{\small $0.46$}};
	\node (TV2) at (1.5,-0.6)   {{\small $0.38$}};
	\node (TV3) at (3,-0.6)   {{\small $0.30$}};
	\node (TV4) at (4.5,-0.6)   {{\small $0.16$}};
	\node (TV5) at (0,2.1)   {{\small $0.46$}};
	\node (TV6) at (1.5,2.1)   {{\small $0.40$}};
	\node (TV7) at (3,2.1)   {{\small $0.32$}};
	\node (TV8) at (4.5,2.1)   {{\small $0.22$}};

    \node (u1) at (6,0.2) [circle,inner sep= 1.5pt,draw] {$u_1$};     
    \node (u2) at (6,1.7) [circle,inner sep= 1.5pt,draw] {$u_2$};
	\node (u11) at (6,2.4)   {{\small $0$}}; 	
    \node (u22) at (6,0.9)   {{\small $0$}};
	\path[] 	
    (u1) [->,thick] edge node[below] {} (n4) 	
    (u2) [->,thick] edge node[below] {} (n8);
	\path[] 	(n3) [->,thick] edge node[below] {} (n2);
	\path[] 	(n6) [->,thick] edge node[below] {} (n5);
	\path[] 	(n6) [->,thick] edge node[below] {} (n1);
	\path[] 	(n2) [->,thick] edge node[below] {} (n6);
	\path[] 	(n7) [->,thick] edge node[below] {} (n6);
	\path[] 	(n3) [->,thick] edge node[below] {} (n6);
	\path[] 	(n3) [->,thick] edge node[below] {} (n7);
	\path[] 	(n8) [->,thick] edge node[below] {} (n7);
	\path[] 	(n8) [->,thick] edge node[below] {} (n3); 
	\path[] 	(n4) [->,thick] edge node[below] {} (n3);

\end{tikzpicture}\caption{The FSN network corresponding to the network $\mathcal{G}_{8}$ in
Figure \ref{fig:8-node-network}.}
\label{fig:8-node-network-FSN}
\end{figure}

\begin{example}
\label{exa:SAN} Consider a SAN on the network $\mathcal{G}_{8}$
shown in Figure \ref{fig:8-node-network}; each agent holds a three-dimensional
state and agents $4$ and $8$ are leaders that are directly influenced
by the homogeneous input $\boldsymbol{u}=(\boldsymbol{u}_{1}^{\top},\thinspace\boldsymbol{u}_{2}^{\top})^{\top}$,
where $\boldsymbol{u}_{1}=\boldsymbol{u}_{2}=(0.7,\thinspace0.8,\thinspace0.9)^{\top}\in\mathbb{R}^{3}$.
The initial states of agents are randomly selected from $[0,1]\times[0,1]\times[0,1]$.
Computing $\boldsymbol{v}_{\text{1}}(L_{B})$ corresponding to the
perturbed Laplacian in this example yields, 
\[
\boldsymbol{v}_{\text{1}}(L_{B})=(0.46,\thinspace0.38,\,0.30,\,0.16,\,0.46,\,0.40,\,0.32,\,0.22){}^{\top}\text{.}
\]
One can observe from Figure \ref{fig:8-node-network} that for each
agent $i\in\mathcal{V}$, there exists a directed path from $\boldsymbol{u}_{1}$
or $\boldsymbol{u}_{2}$ to $i$ such that the entries in $\boldsymbol{v}_{1}(L_{B})$
along this path are monotonically increasing. Therefore, the associated
FSN network according to Definition \ref{def:fsn-network-SAN}, is
as shown in Figure \ref{fig:8-node-network-FSN}. One can observe
from Figure \ref{fig:unsigned-3d-x} that each agent tends to track
the external input directly in the FSN network (see Figure \ref{fig:unsigned-3d-x}b)
instead of aggregating and moving together towards the external input,
this is shown in Figure \ref{fig:unsigned-3d-x}a. 
\end{example}
Theorem \ref{thm:reachability-FSN-SAN} ensures that all agents in
the FSN network of a SAN are influenced by the external inputs, namely,
LF-reachability can be guaranteed. Therefore, a SAN (\ref{eq:unsigned-LF-overall})
can exhibit either consensus or clustering over the corresponding
FSN network, depending on heterogeneity of the external input; see
Lemma \ref{lem:steady-state-unsigned-semi} in Appendix. One can verify
that constructing the FSN network using eigenvectors other than $\boldsymbol{v}_{1}(L_{B})$
do not ensure the LF-reachability according to Definition \ref{def:fsn-network-SAN}.

Inspired by Theorem \ref{thm:reachability-FSN-SAN}, we postulate
that if one reverses the construction of FSN network for SANs (each
agent now follows neighbors whose respective rates of change in state
are relatively faster), the influence of external input exerted on
the network can be weaken or even eliminated. One can refer to the
resulting reduced network as following the faster neighbor (FFN) network.
In this case, agents in the FFN network are not reachable from the
external input. This can be useful when the external input, say, represents
epidemics or rumors, and the network structure is rearranged in a
distributed manner by each agent to attenuate the spreading process.
For example, the FFN network of $\mathcal{G}_{8}$ in Figure \ref{fig:8-node-network}
is shown in Figure \ref{fig:8-node-network-FFN}; in this case, the
influence from external inputs to leaders can be eliminated since
the rate of change in state of external inputs can be viewed as zero.
The trajectory of SAN on FFN network is shown on the right-hand plot
in Figure \ref{fig:unsigned-3d-x}. In FFN networks, the influence
structure is reversed in contrast to FSN network. Therefore, only
leader agents (agents $4$ and $8$) are influenced by external inputs
and as a result, the influence of external inputs on the follower
agents have been eliminated.\textbf{ }In this paper, we shall concentrate
on FSN networks; such networks closely abstract means of enhancing
the spreading process on a network.

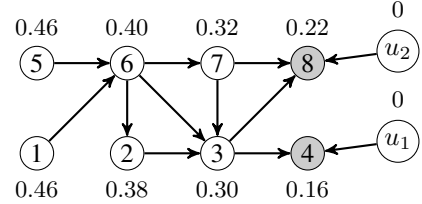
\begin{figure}[tbh]
\centering{}\begin{tikzpicture}[scale=0.8, >=stealth',   pos=.8,  photon/.style={decorate,decoration={snake,post length=1mm}} ]
	\node (n1) at (0,0) [circle,inner sep= 1.5pt,draw] {1};
	\node (n2) at (1.5,0) [circle,inner sep= 1.5pt,draw] {2};
    \node (n3) at (3,0) [circle,inner sep= 1.5pt,draw] {3};
    \node (n4) at (4.5,0) [circle,inner sep= 1.5pt,fill=black!20,draw] {4};
	\node (n5) at (0,1.5) [circle,inner sep= 1.5pt,draw] {5};
	\node (n6) at (1.5,1.5) [circle,inner sep= 1.5pt,draw] {6};
    \node (n7) at (3,1.5) [circle,inner sep= 1.5pt,draw] {7};
    \node (n8) at (4.5,1.5) [circle,inner sep= 1.5pt,fill=black!20,draw] {8};

	\node (TV1) at (0,-0.6)   {{\small $0.46$}};
	\node (TV2) at (1.5,-0.6)   {{\small $0.38$}};
	\node (TV3) at (3,-0.6)   {{\small $0.30$}};
	\node (TV4) at (4.5,-0.6)   {{\small $0.16$}};
	\node (TV5) at (0,2.1)   {{\small $0.46$}};
	\node (TV6) at (1.5,2.1)   {{\small $0.40$}};
	\node (TV7) at (3,2.1)   {{\small $0.32$}};
	\node (TV8) at (4.5,2.1)   {{\small $0.22$}};

    \node (u1) at (6,0.2) [circle,inner sep= 1.5pt,draw] {$u_1$};     
    \node (u2) at (6,1.7) [circle,inner sep= 1.5pt,draw] {$u_2$};
	\node (u11) at (6,2.4)   {{\small $0$}}; 	
    \node (u22) at (6,0.9)   {{\small $0$}};
   \path[] 	
    (u1) [->,thick] edge node[below] {} (n4) 	
    (u2) [->,thick] edge node[below] {} (n8);
	\path[] 	(n3) [<-,thick] edge node[below] {} (n2);
	\path[] 	(n6) [<-,thick] edge node[below] {} (n5);
	\path[] 	(n6) [<-,thick] edge node[below] {} (n1);
	\path[] 	(n2) [<-,thick] edge node[below] {} (n6);
	\path[] 	(n7) [<-,thick] edge node[below] {} (n6);
	\path[] 	(n3) [<-,thick] edge node[below] {} (n6);
	\path[] 	(n3) [<-,thick] edge node[below] {} (n7);
	\path[] 	(n8) [<-,thick] edge node[below] {} (n7);
	\path[] 	(n8) [<-,thick] edge node[below] {} (n3); 
	\path[] 	(n4) [<-,thick] edge node[below] {} (n3);

\end{tikzpicture}\caption{The FFN network corresponding to the network $\mathcal{G}_{8}$ in
Figure \ref{fig:8-node-network}. }
\label{fig:8-node-network-FFN}
\end{figure}

\begin{figure*}[tbh]
\begin{centering}
\includegraphics[width=5.5cm]{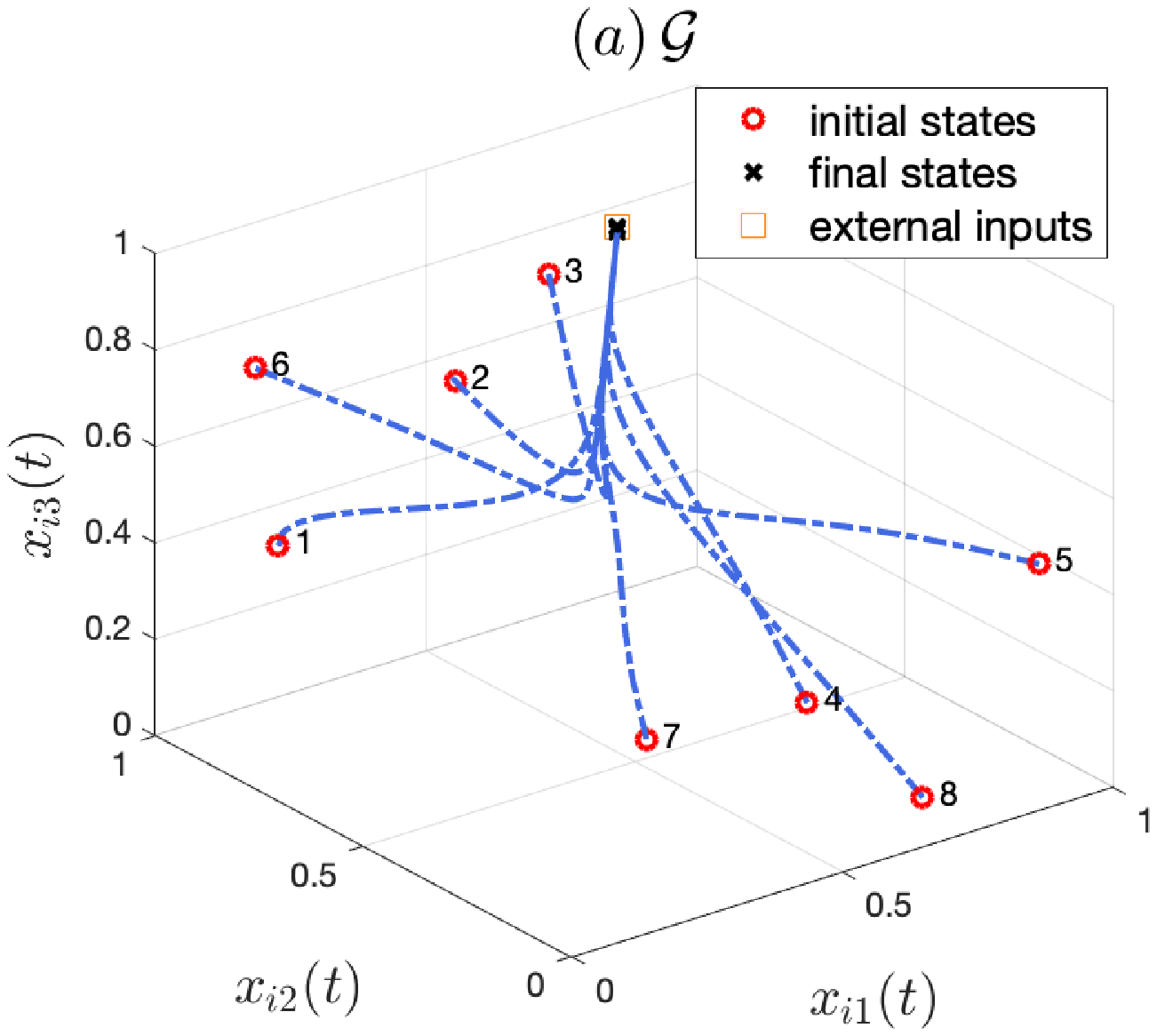}\,\,\includegraphics[width=5.5cm]{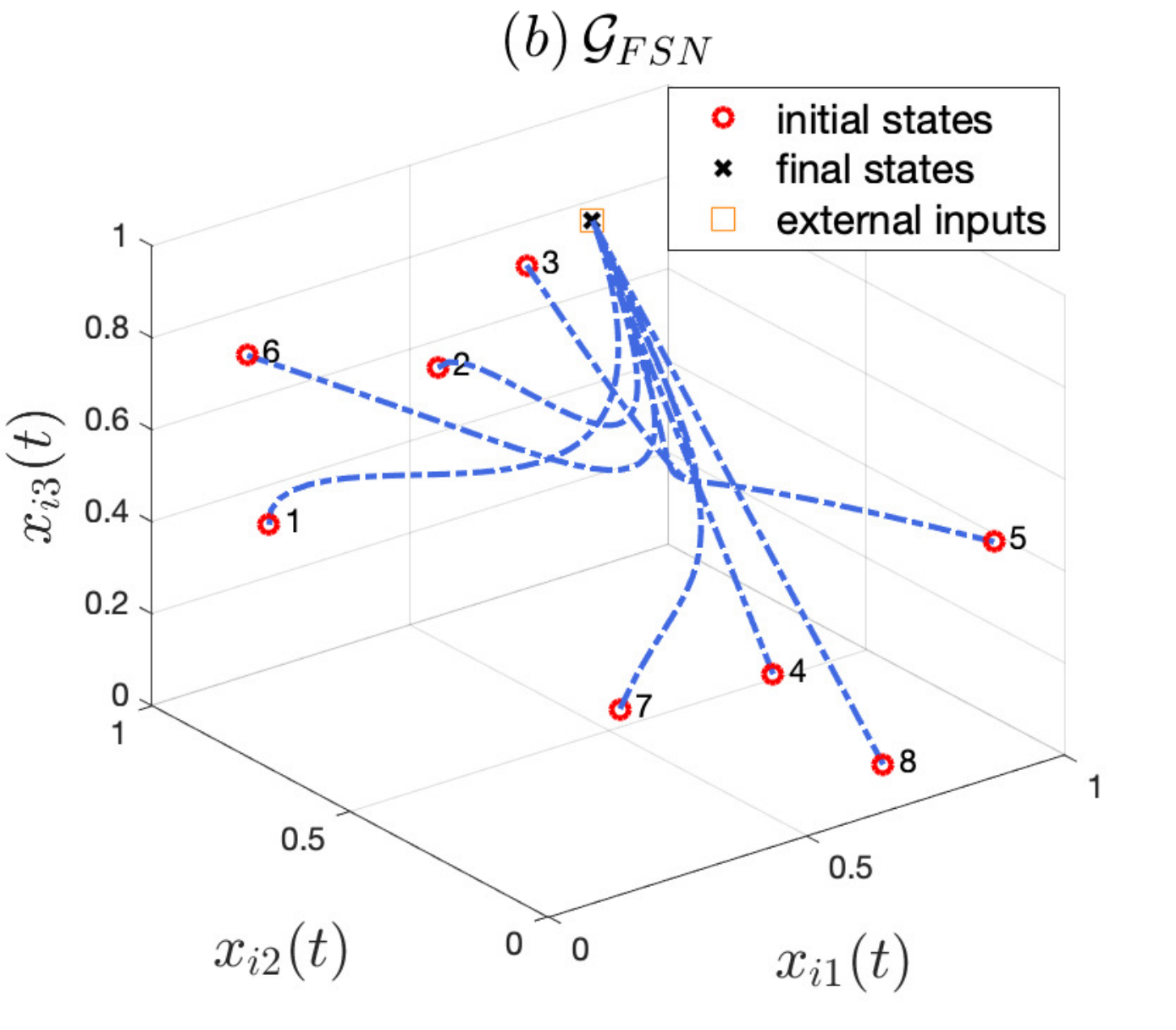}\,\,\includegraphics[width=5.5cm]{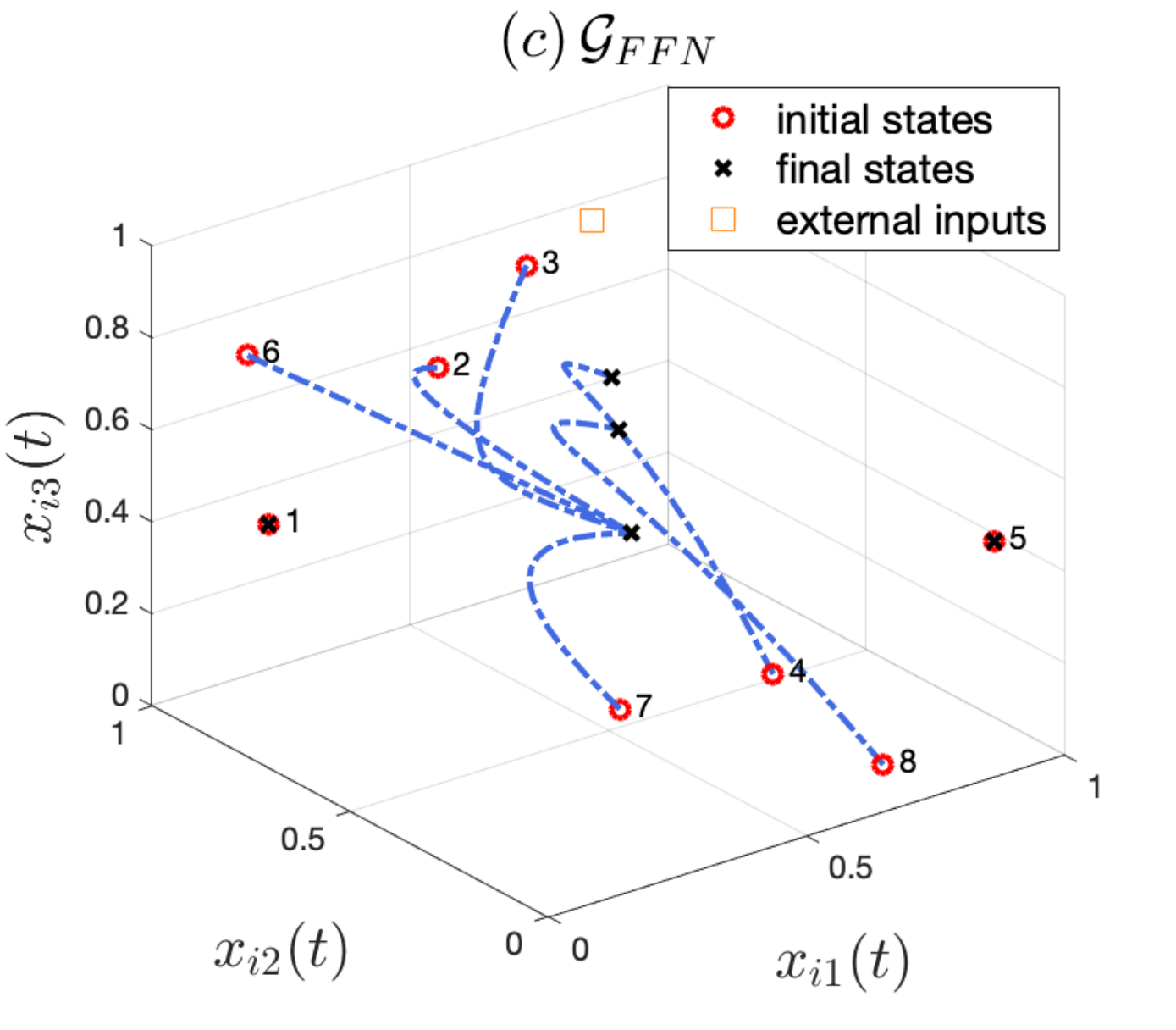}
\par\end{centering}
\centering{}\caption{State trajectories of agents in SAN on the network $\mathcal{G}_{8}$
shown in Figure \ref{fig:8-node-network} as well as its associated
FSN network (Figure \ref{fig:8-node-network-FSN}) and FFN network
(Figure \ref{fig:8-node-network-FFN}), respectively.}
\label{fig:unsigned-3d-x}
\end{figure*}

\subsection{Convergence Rate Enhancement }

In order to evaluate the performance of neighbor selection based on
$\boldsymbol{v}_{\text{1}}(L_{B})$, we now proceed to examine the
convergence rate of SANs on the resultant FSN networks. Note that
the smallest non-zero eigenvalue of the perturbed Laplacian of a SAN
characterizes the convergence rate of the multi-agent system towards
its steady-state, either consensus or clustering \cite{clark2018maximizing,xia2017analysis,pirani2016smallest}.
We provide the following result on the convergence rate of SAN on
connected unsigned networks and the corresponding FSN networks. 
\begin{thm}
\label{thm:convergence-rate-FSN-SAN} Let $\bar{\mathcal{G}}=(\mathcal{V},\mathcal{\bar{E}},\bar{W})$
denote the FSN network of a SAN $\mathcal{G}=(\mathcal{V},\mathcal{E},W)$
with the input matrix $B$. Then 
\[
\lambda_{1}(L_{B}(\bar{\mathcal{G}}))\ge\lambda_{1}(L_{B}(\mathcal{G})),
\]
where equality holds only when all agents are leaders.
\end{thm}
\begin{proof}
Denote the perturbed Laplacian matrix $L_{B}(\bar{\mathcal{G}})$
and $L_{B}(\mathcal{G})$ as $\bar{L}_{B}$ and $L_{B}$, respectively.
According to Definition \ref{def:fsn-network-SAN}, the FSN network
$\bar{\mathcal{G}}$ is a directed acyclic network. Therefore, one
can relabel the agents in $\bar{\mathcal{G}}$ such that $\bar{W}$
is a lower triangular matrix, implying that $\bar{L}_{B}$ is also
lower triangular. Thus, the eigenvalues of $\bar{L}_{B}$ are exactly
the entries on its diagonal. 

Note that each agent in the FSN network has at least one in-degree
neighbor, namely, the diagonal entries, satisfying that $[\bar{L}_{B}]_{ii}\geq1$
for all $i\in\mathcal{V}$, which in turn, implies that $\lambda_{1}(\bar{L}_{B})\geq1.$
In particular, if all agents are leaders ($\text{{\bf diag}}(B\mathds{1}_{m})=I$),
the eigenvector corresponding to $\lambda_{1}(L_{B})$ is $a_{0}\mathds{1}_{n}$
($a_{0}\in\mathbb{R}$). Then according to Definition \ref{def:fsn-network-SAN},
all the edges in the graph $\mathcal{G}$ will be eliminated. Therefore,
$\lambda_{1}(\bar{L}_{B})=1$. Recall that $L_{B}=L+\text{{\bf diag}}(B\mathds{1}_{m}).$
Applying Weyl theorem (\cite[Theorem 4.3.1,  p.239]{horn2012matrix}),
one has, 
\begin{equation}
\lambda_{1}(L_{B})\leq\lambda_{1}(L)+\lambda_{n}(\text{{\bf diag}}(B\mathds{1}_{m}));
\end{equation}
due to the fact $\lambda_{1}(L)=0$ and $\lambda_{n}(\text{{\bf diag}}(B\mathds{1}_{m}))=1$,
it follows that $\lambda_{1}(L_{B})\leq1.$

On the one hand, if $\text{{\bf diag}}(B\mathds{1}_{m})=I$, again
using Weyl theorem, it follows that,

\begin{equation}
\lambda_{1}(L_{B})\geq\lambda_{1}(L)+\lambda_{1}(I)=1;
\end{equation}
hence, $\lambda_{1}(L_{B})=1.$ Now suppose that $\text{{\bf diag}}(B\boldsymbol{1}_{m})\neq I$.
Let
\begin{equation}
L_{B}=L+I+\triangle,\label{eq:proof-theorem-2-1}
\end{equation}
where $\triangle\in\mathbb{R}^{n\times n}$ is a non-zero diagonal
matrix whose diagonal entries are either $-1$ or $0$. In fact, (\ref{eq:proof-theorem-2-1})
produces all possible perturbed Laplacians apart from the case that
all agents are leaders. Without loss of generality, we choose $\triangle=\text{{\bf diag}}(-1,0,\ldots,0)^{\top}.$
Then, by applying Weyl theorem one more time, we have,
\begin{equation}
\lambda_{1}(L+I+\triangle)\leq\lambda_{1}(L+I)+\lambda_{n}(\triangle)=1.\label{eq:proof-theorem-2-2}
\end{equation}

According to $(L+I)(a\mathds{1}_{n})=a\mathds{1}_{n}$, where $a\in\mathbb{R}$,
it follows that $\text{{\bf span}}\left\{ \mathds{1}_{n}\right\} $
is an eigenspace of the matrix $L+I$ corresponding to the eigenvalue
$\lambda_{1}(L+I)$. 

For the matrix $\triangle$ in (\ref{eq:proof-theorem-2-2}), assume
that there exists $a_{0}\in\mathbb{R}$ such that,
\begin{equation}
\triangle(a_{0}\mathds{1}_{n})=\lambda_{n}(\triangle)a_{0}\mathds{1}_{n};
\end{equation}
as $\triangle(a_{0}\mathds{1}_{n})=(-a_{0},0,\ldots,0)^{\top}$ and
$\lambda_{n}(\triangle)a_{0}\mathds{1}_{n}=(0,0,\ldots,0)^{\top}$,
one has $a_{0}=0$. 

Thus, there does not exist a common non-zero eigenvector corresponding
to $\lambda_{1}(L+I+\triangle)$, $\lambda_{1}(L+I)$ and $\lambda_{n}(\triangle)$,
respectively. According to Weyl theorem, one has,
\begin{equation}
\lambda_{1}(L+I+\triangle)<1.
\end{equation}

Thus, one can conclude that $\lambda_{1}(\bar{L}_{B})>\lambda_{1}(L_{B})$
if $\text{{\bf diag}}(B\mathds{1}_{m})\neq I$ and $\lambda_{1}(\bar{L}_{B})=\lambda_{1}(L_{B})=1$
when $\text{{\bf diag}}(B\mathds{1}_{m})=I$. 
\end{proof}
Theorem \ref{thm:convergence-rate-FSN-SAN} provides theoretical guarantees
on the convergence rate of FSN network $\bar{\mathcal{G}}$ as compared
with the original network $\mathcal{G}$. Let us continue to employ
Example \ref{exa:SAN} to demonstrate the convergence rate of SAN
on the original network $\mathcal{G}$ and its related FSN network
$\bar{\mathcal{G}}$. The convergence rate comparison of SAN in Example
\ref{exa:SAN} on both original network and the associated FSN network
is demonstrated in Figure \ref{fig:trajectory-original-8-node}. In
this case, the convergence rate of SAN on the original network $\mathcal{G}$
and FSN network $\mathcal{G}_{FSN}$ are $\lambda_{\text{1}}(L_{B}(\mathcal{G}))=0.1414$
and $\lambda_{\text{1}}(L_{B}(\mathcal{G}_{FSN}))=1$, respectively.

\begin{figure}[tbh]
\begin{centering}
\includegraphics[width=9cm]{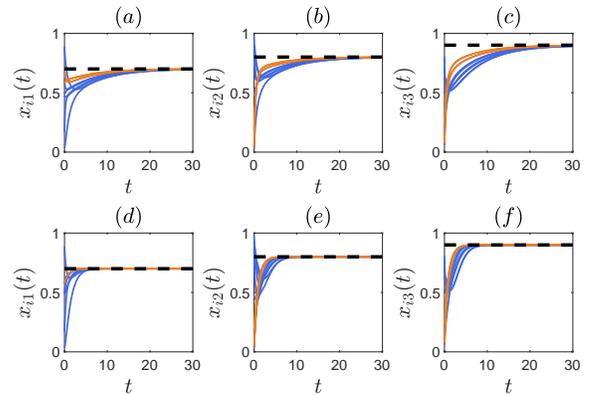}
\par\end{centering}
\caption{State trajectories of agents in the SAN (\ref{eq:unsigned-LF-overall})
on the network in Figure \ref{fig:8-node-network} ((a)-(c)). State
trajectories of agents in the SAN (\ref{eq:unsigned-LF-overall})
on the FSN network in Figure \ref{fig:8-node-network-FSN} ((d)-(f)).
The orange and blue lines in each plot are  trajectories of leader
and follower agents, respectively. The dotted line in each panel represents
external input.}
\label{fig:trajectory-original-8-node}
\end{figure}

\subsection{Neighbor Selection in SANs\label{subsec:NS-SAN}}

So far, we have presented a neighbor selection framework with guaranteed
performance using the eigenvector of perturbed Laplacian. However,
this eigenvector is a network-level quantity, hindering the direct
applicability of this setup for large-scale networks. For such networks,
it is desirable that decision-making relies only on local observations
\cite{li2014designing}. 

In this section, we establish a quantitative link between the Laplacian
eigenvector and the relative rate of change in the state of neighboring
agents, referred to as the relative tempo. Using this approach, we
connect the global property of the network to a locally measurable
quantity-- as such, we are able to propose a fully distributed neighbor
selection algorithm. In order to simplify our derivation, we first
introduce the so-called selection matrix.

The selection matrix of a subset of agents $\mathcal{V}^{\prime}=\left\{ i_{1},\ldots,i_{s}\right\} \subset\mathcal{V}$
is defined as $\phi(\mathcal{V}^{\prime})=(\boldsymbol{e}_{i_{1}},\cdots,\boldsymbol{e}_{i_{s}})^{\top}\in\mathbb{R}^{s\times n}.$
We now proceed to introduce the notion of relative tempo, characterizing
the steady-state of the relative rate of change in state between two
subsets of agents. 
\begin{defn}
\label{def:relative-tempo}  Let $\mathcal{V}_{1}\subset\mathcal{V}$
and $\mathcal{V}_{2}\subset\mathcal{V}$ be two subsets of agents
in multi-agent network (\ref{eq:unsigned-LF-overall}) (or (\ref{eq:consensus-overall})).
Then the relative tempo between agents in $\mathcal{V}_{1}$ and $\mathcal{V}_{2}$
is defined as the limiting ratio,
\begin{equation}
{\color{blue}\mathbb{L}}(\mathcal{V}_{1},\mathcal{V}_{2})=\lim_{t\rightarrow\infty}\frac{\|\phi(\mathcal{V}_{1})\otimes I_{d}\dot{\boldsymbol{x}}(t)\|}{\|\phi(\mathcal{V}_{2})\otimes I_{d}\dot{\boldsymbol{x}}(t)\|},\label{eq:relative-tempo-definition}
\end{equation}
where $\phi(\mathcal{V}_{1})$ and $\phi(\mathcal{V}_{2})$ are selection
matrices associated with $\mathcal{V}_{1}$ and $\mathcal{V}_{2}$,
respectively. 
\end{defn}
The relative tempo in Definition \ref{def:relative-tempo} was initially
examined in \cite{HaibinAcc14}, characterizing relative influence
of agents in consensus-type networks, and subsequently being employed
to construct a centrality measure that can be inferred from network
data \cite{shao2017inferring}. This paper provides a more systematic
treatment for the application of relative tempo in the distributed
neighbor selection problem. As we shall see subsequently, the limit
in (\ref{eq:relative-tempo-definition}) exists, implying that the
relative tempo is well-defined. We now proceed to formally provide
a quantitative connection between relative tempo and the Laplacian
eigenvector. 
\begin{thm}
\label{thm:relative-tempo-SAN}Let $\mathcal{V}_{1}\subset\mathcal{V}$
and $\mathcal{V}_{2}\subset\mathcal{V}$ be two subsets of agents
in the SAN (\ref{eq:unsigned-LF-overall}). Then
\[
{\color{blue}\mathbb{L}}(\mathcal{V}_{1},\mathcal{V}_{2})=\frac{\|\phi(\mathcal{V}_{1})\boldsymbol{v_{1}}(L_{B})\|}{\|\phi(\mathcal{V}_{2})\boldsymbol{v_{1}}(L_{B})\|}.
\]
\end{thm}
\begin{proof}
Refer to the Appendix
\end{proof}
\begin{rem}
Theorem \ref{thm:relative-tempo-SAN} provides a quantitative connection
between the relative tempo (constructed from local observations of
each agent) and the Laplacian eigenvector of the underlying network.
According to Theorem \ref{thm:reachability-FSN-SAN} and Theorem \ref{thm:convergence-rate-FSN-SAN},
such a connection enables a distributed implementation of neighbor
selection for enhancing the convergence rate of the network. We provide
an example to illustrate Theorem \ref{thm:relative-tempo-SAN}.
\end{rem}
\begin{example}
Consider the following quantity
\begin{equation}
g_{ij}(t)=\frac{\|\dot{\boldsymbol{x}}_{i}(t)\|}{\|\dot{\boldsymbol{x}}_{j}(t)\|},\thinspace i\in\mathcal{V},j\in\mathcal{N}_{i},\label{eq:g_i}
\end{equation}
which satisfies ${\displaystyle \lim_{t\rightarrow\infty}}g_{ij}(t)={\color{blue}\mathbb{L}}(i,j)$
(by Definition \ref{def:relative-tempo}). Let us continue to examine
Example \ref{exa:SAN}. The trajectories of $g_{ij}(t)$ for $i=7$
and $j\in\{3,6,8\}$ are shown in Figure \ref{fig:g-trajectory-semi-unsigned}.
The steady-states of $g_{ij}(t)$ are archived at around $t=10$,
particularly, $g_{73}(10)=1.057$, $g_{76}(10)=0.8123$ and $g_{78}(10)=1.47$,
respectively. In the meanwhile, one has $\boldsymbol{v}_{\text{1}}(L_{B})_{73}=1.0577$,
$\boldsymbol{v}_{\text{1}}(L_{B})_{76}=0.8113$ and $\boldsymbol{v}_{\text{1}}(L_{B})_{78}=1.4694$.
Note that such correspondences are sufficient for the construction
of the associated FSN network (shown in Figure \ref{fig:8-node-network-FSN})
using Definition \ref{def:fsn-network-SAN}. 
\begin{figure}[tbh]
\begin{centering}
\includegraphics[width=7cm]{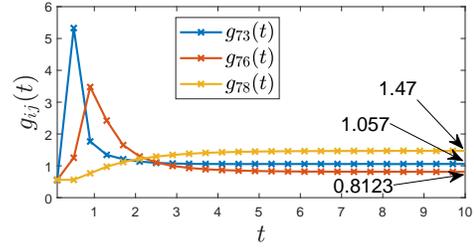}
\par\end{centering}
\caption{Trajectories of $g_{ij}(t)$ for agent $i=7$ and its neighbors $j\in\{3,6,8\}$
in the SAN shown in Figure \ref{fig:8-node-network}.}

\label{fig:g-trajectory-semi-unsigned}
\end{figure}
\end{example}
\begin{rem}
Examining the convergence rate of agents' states in the original network
shown in Figure \ref{fig:trajectory-original-8-node} and that of
$g_{ij}(t)$ in Figure \ref{fig:g-trajectory-semi-unsigned}, one
observes that the SAN (\ref{eq:unsigned-LF-overall}) achieves an
\textquotedbl ordered\textquotedbl{} state characterized by the relative
tempo, prior to the final consensus. 
\end{rem}
By Definition \ref{def:fsn-network-SAN} and Theorem \ref{thm:relative-tempo-SAN},
the reduced neighbor set for each agent for constructing the FSN network
 can be defined as follows. 
\begin{defn}
The reduced neighbor set for agent $i\text{\ensuremath{\in\mathcal{V}}}$
to construct the associated FSN network for SAN (\ref{eq:unsigned-LF-overall})
is defined as
\begin{equation}
\mathcal{N}_{i}^{\text{FSN}}=\left\{ j\in\mathcal{N}_{i}\mid\boldsymbol{v}_{\text{1}}(L_{B})_{ij}>1\right\} =\left\{ j\in\mathcal{N}_{i}\mid{\color{blue}\mathbb{L}}(i,j)>1\right\} .
\end{equation}
\end{defn}
According to Definition \ref{def:relative-tempo}, if ${\color{blue}\mathbb{L}}(i,j)>1$
for a pair of neighboring agents $i,j\in\mathcal{V}$, then the state
of agent $i$ evolves towards the external input in a relatively faster
rate than that of agent $j$. Therefore, agents in a SAN are involved
in a sort of hierarchy encoded in $\boldsymbol{v}_{1}(L_{B})$ according
to Theorem \ref{thm:relative-tempo-SAN}. As such, each agent can
select a specific group of neighbors to interact with for a given
task. The main insight from our discussion is that $\boldsymbol{v}_{\text{1}}(L_{B})_{ij}$
can be estimated for agent $i\in\mathcal{V}$ and $j\in\mathcal{N}_{i}$
via only local measurements, and the obtained reduced network exhibits
LF-reachability property, as stated by Theorem \ref{thm:reachability-FSN-SAN}.

To end the discussion on SANs, we provide the following algorithm
to summarize the flowchart of distributed neighbor selection process
for constructing FSN networks. For the algorithm implementation, we
choose the sampling step size $\delta>0$ to discretize agent state
evolution as $\ensuremath{\tilde{x}_{i}(k)=x_{i}(t_{0}+k\delta)}$
for all $i\in\mathcal{V}$, where $t_{0}=0$ and $k=0,1,\ldots$. 

\begin{algorithm}[h]
\begin{algorithmic}[1] 
\Require{} 
\State{set $k=1$}
\For {each agent $i\in\mathcal{V}$}
\State{choose the termination threshold $\varepsilon_i>0$}
\State{receives $\tilde{x}_{i}(0)$ and $\tilde{x}_{i}(1)$ from $j\in\mathcal{N}_i$}  
\State{computes $g_{ij}(k)=\frac{\|\tilde{x}_{i}(k)-\tilde{x}_{i}(k-1)\|}{\|\tilde{x}_{j}(k)-\tilde{x}_{j}(k-1)\|}$}
\EndFor
\Ensure{}
\Repeat{}
\State{set $k=k+1$}
\For {each agent $i\in\mathcal{V}$}
\State{receives $\tilde{x}_{i}(k)$ from $j\in\mathcal{N}_i$}  
\State{computes $g_{ij}(k)=\frac{\|\tilde{x}_{i}(k)-\tilde{x}_{i}(k-1)\|}{\|\tilde{x}_{j}(k)-\tilde{x}_{j}(k-1)\|}$}
\EndFor
\Until{$\|g_{ij}(k)-g_{ij}(k-1)\|<\varepsilon_i, \forall j\in\mathcal{N}_i$}
\State{$\ensuremath{\bar{w}_{ij}=\begin{cases} w_{ij}, & \ensuremath{g_{ij}(k)>1,}\\ 0, & \ensuremath{g_{ij}(k)\le 1}. \end{cases}}$} 
\end{algorithmic}

\caption{Distributed neighbor selection for SANs.}
\label{algorightm 1}
\end{algorithm}

Note from Algorithm \ref{algorightm 1} that each agent only uses
local accessible state information to construct FSN networks.

\section{Fully-autonomous Networks\label{sec:FAN}}

In this section, we proceed to investigate parallel results for FANs-
the corresponding analysis turns out to be more intricate than those
for SANs. 

\subsection{Reachability Analysis}

Recall that for the eigenvector associated with perturbed Laplacian
matrix $L_{B}$ in SANs (\ref{eq:unsigned-LF-overall}), all elements
of $\boldsymbol{v}_{\text{1}}(L_{B})$ have the same sign. However,
in the case of FANs, the entries in eigenvectors of graph Laplacian
can be positive, negative or equal to zero; this is also valid for
the Fiedler vector $\boldsymbol{v}_{\text{2}}(L)$, the eigenvector
corresponding to the second smallest eigenvalue of graph Laplacian
\cite{fiedler1973algebraic,Fiedler1975,shao2015fiedler}. This situation
renders the extension of the aforementioned neighbor selection framework
-from SANs to FANs- non-trivial. 

In this section, we shall first examine the property of Laplacian
eigenvectors related to SANs, specifically the Fiedler vector $\boldsymbol{v}_{\text{2}}(L)$,
and then proceed to provide the neighbor selection algorithm to construct
the FSN network of FANs for fast convergence. In the following discussions,
we shall refer to a node corresponding to positive, negative or zero
entry in $\boldsymbol{v}_{\text{2}}(L)$ as a positive node, negative
node and zero node, respectively.

In FANs, the structural properties of the network turn out to be critical
in the analysis and design of the neighbor selection algorithm. We
introduce the following results related to the block decomposition
of a graph \cite{harary1966block,Fiedler1975}.

A cut node $i\in\mathcal{V}$ of a graph $\mathcal{G}=(\mathcal{V},\mathcal{E},W)$
is a node such that $\mathcal{G}-\left\{ i\right\} $ is disconnected.
A block of a graph $\mathcal{G}$ is a maximal connected subgraph
of $\mathcal{G}$ with no cut nodes. Two blocks of $\mathcal{G}$
are the neighboring blocks if they are connected via a cut node. Consider
a connected graph $\mathcal{G}$ with blocks $\left\{ B_{i}\right\} $
and cut nodes $\left\{ c_{j}\right\} $, where $i,j\in\mathbb{Z}_{+}$.
The block-cut graph of $\mathcal{G}$, denoted by $\mathcal{B}(\mathcal{G})$,
is defined as the graph with node set composed of blocks and cut nodes,
namely, $\mathcal{V}(\mathcal{B}(\mathcal{G}))=\left\{ B_{i}\right\} \cup\left\{ c_{j}\right\} $,
where two nodes are adjacent if one corresponds to a block $B_{i}$
and the other to a cut node such that $c_{j}\in B_{i}$. 
\begin{lem}
\cite{harary1966block} Let $\mathcal{G}=(\mathcal{V},\mathcal{E},W)$
be a connected graph. Then the block-cut graph of $\mathcal{G}$ is
a tree. 
\end{lem}
Here, we provide an example to illustrate the block-cut tree associated
with a connected graph.
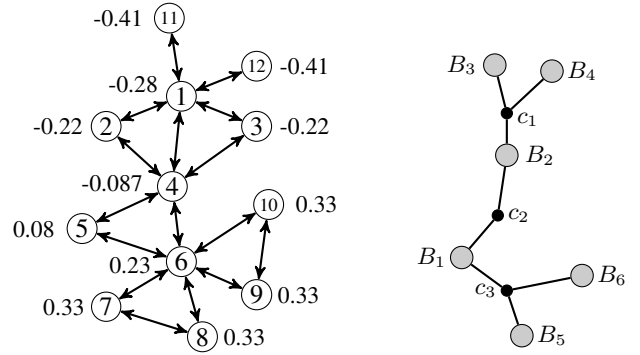
\begin{figure}[tbh]
\begin{centering}
\begin{tikzpicture}[scale=0.8, >=stealth',   pos=.8,  photon/.style={decorate,decoration={snake,post length=1mm}} ]
	\node (n1) at (0.75,1) [circle,inner sep= 1pt,draw] {1};
	\node (n2) at (-0.5,0.5) [circle,inner sep= 1pt,draw] {2};
    \node (n3) at (2,0.5) [circle,inner sep= 1pt,draw] {3};
    \node (n4) at (0.6,-0.5) [circle,inner sep= 1pt,draw] {4};
	\node (n5) at (-0.9,-1.2) [circle,inner sep= 1pt,draw] {5};
	\node (n6) at (0.75,-1.75) [circle,inner sep= 1pt,draw] {6};
    \node (n7) at (-0.5,-2.5) [circle,inner sep= 1pt,draw] {7};
    \node (n8) at (1.1,-3) [circle,inner sep= 1pt,draw] {8};
    \node (n9) at (2,-2.3) [circle,inner sep= 1pt,draw] {9};
    \node (n10) at (2.2,-0.8) [circle,inner sep= 1pt,draw] {{\scriptsize 10}};
    \node (n11) at (0.55,2.3) [circle,inner sep= 1pt,draw] {{\scriptsize 11}};
    \node (n12) at (2,1.5) [circle,inner sep= 1pt,draw] {{\scriptsize 12}};

	\node (TV1) at (-0.05,1.2)   {{\small -0.28}};
	\node (TV2) at (-1.3,0.5)   {{\small -0.22}};
	\node (TV3) at (2.8,0.5)   {{\small -0.22}};
	\node (TV4) at (-0.4,-0.45)   {{\small -0.087}};
	\node (TV5) at (-1.7,-1.2)   {{\small 0.08}};
	\node (TV6) at (-0.1,-1.8)   {{\small 0.23}};
	\node (TV7) at (-1.2,-2.5)   {{\small 0.33}};
	\node (TV8) at (1.8,-3)   {{\small 0.33}};
    \node (TV9) at (2.7,-2.3)   {{\small 0.33}};
    \node (TV10) at (3,-0.8)   {{\small 0.33}};
    \node (TV11) at (-0.3,2.3)   {{\small -0.41}};
    \node (TV12) at (2.8,1.5)   {{\small -0.41}};

	\path[]
	(n1) [<->,thick, right=8] edge node[below] {} (n2);
    \path[]
	(n1) [<->,thick, right=8] edge node[below] {} (n3);

    \path[]
    (n1) [<->,thick, right=8] edge node[below] {} (n4);
    \path[]
    (n1) [<->,thick, right=8] edge node[below] {} (n11);
    \path[]
    (n1) [<->,thick, right=8] edge node[below] {} (n12);

	\path[]
	(n2) [<->,thick, right=8] edge node[below] {} (n4);
	\path[]
	(n3) [<->,thick, right=8] edge node[below] {} (n4);

	\path[]
	(n4) [<->,thick, right=8] edge node[below] {} (n5);
    \path[]
    (n4) [<->,thick, right=8] edge node[below] {} (n6);
	\path[]
	(n5) [<->,thick, right=8] edge node[below] {} (n6);

	\path[]
	(n6) [<->,thick, right=8] edge node[below] {} (n7);
    \path[]
	(n6) [<->,thick, right=8] edge node[below] {} (n8);
    \path[]
    (n6) [<->,thick, right=8] edge node[below] {} (n9);
    \path[]
    (n6) [<->,thick, right=8] edge node[below] {} (n10);

	\path[]
	(n7) [<->,thick, right=8] edge node[below] {} (n8);
	\path[]
	(n9) [<->,thick, right=8] edge node[below] {} (n10);

\end{tikzpicture}\,\,\,\,\,\,\,\,\,\,\,\,\,\,\begin{tikzpicture}[scale=0.8, >=stealth',   pos=.8,  photon/.style={decorate,decoration={snake,post length=1mm}} ]
	\node (n1) at (0.75,1.2) [circle,inner sep= 1.5pt,fill=black,draw] {};
	\node (n2) at (0.75,0.5) [circle,inner sep= 3pt,fill=black!20,draw] {};
    \node (n4) at (0.6,-0.5) [circle,inner sep= 1.5pt,fill=black,draw] {};
	\node (n5) at (0,-1.2) [circle,inner sep= 3pt,fill=black!20,draw] {};
	\node (n6) at (0.75,-1.75) [circle,inner sep= 1.5pt,fill=black,draw] {};
    
    \node (n8) at (1,-2.5) [circle,inner sep= 3pt,fill=black!20,draw] {};
    
    \node (n10) at (2,-1.5) [circle,inner sep= 3pt,fill=black!20,draw] {};
    \node (n11) at (0.55,2) [circle,inner sep= 3pt,fill=black!20,draw] {};
    \node (n12) at (1.5,1.9) [circle,inner sep= 3pt,fill=black!20,draw] {};

	\node (TV1) at (1.1,1.1)   {{\small $c_1$}};
	\node (TV2) at (1.3,0.5)   {{\small $B_2$}};
	\node (TV3) at (0.95,-0.5)   {{\small $c_2$}};
	\node (TV4) at (-0.5,-1.2)   {{\small $B_1$}};
	\node (TV5) at (0.4,-1.75)   {{\small $c_3$}};
	\node (TV6) at (1.5,-2.5)   {{\small $B_5$}};
	\node (TV7) at (2.5,-1.5)   {{\small $B_6$}};
	\node (TV8) at (0,2)   {{\small $B_3$}};
    \node (TV9) at (2,1.9)   {{\small $B_4$}};

    \path[]
    (n1) [-,thick, right=8] edge node[below] {} (n2);
    \path[]
    (n1) [-,thick, right=8] edge node[below] {} (n11);
    \path[]
    (n1) [-,thick, right=8] edge node[below] {} (n12);

	\path[]
	(n5) [-,thick, right=8] edge node[below] {} (n4);
    \path[]
	(n2) [-,thick, right=8] edge node[below] {} (n4);

	\path[]
	(n5) [-,thick, right=8] edge node[below] {} (n6);

    \path[]
	(n6) [-,thick, right=8] edge node[below] {} (n8);
   
    \path[]
    (n6) [-,thick, right=8] edge node[below] {} (n10);

\end{tikzpicture}
\par\end{centering}
\caption{An unsigned network $\mathcal{G}_{12}$ with $12$ nodes (left) and
the associated block-cut tree (right). The entry in $\boldsymbol{v}_{\text{2}}(L)$
corresponding to each agent is shown close to each node. The block
decomposition of network $\mathcal{G}_{12}$ is shown in the right
where the black nodes represent cut nodes in $\mathcal{G}_{12}$ and
the grey nodes represent blocks in $\mathcal{G}_{12}$.}

\label{fig:12-node-network}
\end{figure}

\begin{example}
\label{block-decomposition-graph}Consider a FAN $\mathcal{G}_{12}$
shown on the left-hand plot of Figure \ref{fig:12-node-network}.
There are six blocks in $\mathcal{G}_{12}$, that is, 
\begin{figure}[tbh]
\centering{}%
\begin{tabular}{ll}
$B_{1}=\mathcal{G}\left(\left\{ 4,5,6\right\} \right),$ & $B_{2}=\mathcal{G}\left(\left\{ 1,2,3,4\right\} \right),$\tabularnewline
$B_{3}=\mathcal{G}\left(\left\{ 1,11\right\} \right),$ & $B_{4}=\mathcal{G}\left(\left\{ 1,12\right\} \right),$\tabularnewline
$B_{5}=\mathcal{G}\left(\left\{ 6,7,8\right\} \right),$ & $B_{6}=\mathcal{G}\left(\left\{ 6,9,10\right\} \right).$\tabularnewline
\end{tabular}
\end{figure}

There are three cut nodes $c_{1}=\left\{ 1\right\} ,\,c_{2}=\left\{ 4\right\} ,\,c_{3}=\left\{ 6\right\} .$
The block-cut tree of $\mathcal{G}_{12}$ is shown on the right panel
in Figure \ref{fig:12-node-network}. Blocks $B_{1}$ and $B_{2}$
are neighboring blocks since they are connected via cut node $c_{2}=\left\{ 4\right\} $.
\end{example}
The following result reveals the monotonicity property of the entries
in $\boldsymbol{v}_{\text{2}}(L)$ along certain paths in block-cut
tree of a graph, which will subsequently be used for constructing
the FSN network associated with FANs.
\begin{lem}
\label{lem:monotonicity-Fiedler} \cite[Theorem 3.12]{Fiedler1975}
Let \textup{$\mathcal{G}$ be a connected graph with Laplacian matrix
$L$; let }$\lambda_{\text{2}}(L)$ and $\boldsymbol{v}_{\text{2}}(L)$
be the second smallest eigenvalue of $L$ and the corresponding eigenvector,
respectively. Then exactly one of the following two cases occurs:

Case 1. There is a single block $B_{0}$ in $\mathcal{G}$ which contains
both positive and negative nodes. Each other block has either positive
nodes only, or negative nodes only, or zero nodes only. Every path
$P$ starting in $B_{0}$ and containing just one node $k$ in $B_{0}$
has the property that the entries in $\boldsymbol{v}_{\text{2}}(L)$
corresponding to cut nodes contained in $P$ form either an increasing,
or decreasing, or a zero sequence along this path according to whether
$[\boldsymbol{v}_{\text{2}}(L)]_{k}>0$, $[\boldsymbol{v}_{\text{2}}(L)]_{k}<0$
or $[\boldsymbol{v}_{\text{2}}(L)]_{k}=0$; in the last case all nodes
in $P$ are zero nodes. 

Case 2. No block of $\mathcal{G}$ contains both positive and negative
nodes. There exists a single zero cut node with a non-zero node neighbor.
Each block (with the exception of that zero cut node) has either positive
nodes only, or negative nodes only, or zero nodes only. Every path
$P$ starting in that zero cut node has the property that the entries
in $\boldsymbol{v}_{\text{2}}(L)$ corresponding to cut nodes contained
in $P$ form either an increasing, or decreasing, or a zero sequence
along this path and in the last case all nodes in $P$ are zero nodes.
Every path containing both positive and negative nodes passes through
that zero cut node.
\end{lem}
As an example, the FAN $\mathcal{G}_{12}$ in Figure \ref{fig:12-node-network}
satisfies Case 1 in the Lemma \ref{lem:monotonicity-Fiedler}. In
the following discussions, we will refer to the block $B_{0}$ in
Case 1 and the zero cut node in Case 2 in the Lemma \ref{lem:monotonicity-Fiedler}
as core block and core node, respectively; we will refer to a block
having only positive, negative, or zero nodes (with the exception
of that core node) as a positive block, negative block, and zero block,
respectively. We are now ready to present the construction of FSN
networks for FANs.
\begin{defn}[\textbf{FSN Network of FANs}]
\label{def:fsn-tempo-network-autonomous} Let $\{B_{1},\ldots,B_{r}\}$
be the block decomposition of an unsigned network $\mathcal{G}=(\mathcal{V},\mathcal{E},W)$
where $r\in\mathbb{Z}_{+}$. The FSN network of $\mathcal{G}$ is
its subgraph $\bar{\mathcal{G}}=(\bar{\mathcal{\mathcal{V}}},\bar{\mathcal{E}},\bar{W})$
with node set $\bar{\mathcal{\mathcal{V}}}=\mathcal{V}$, edge set
$\mathcal{\bar{E}}\subseteq\mathcal{E}$ and adjacency matrix $\bar{W}=(\bar{w}_{ij})\in\mathbb{R}^{n\times n}$
that satisfies, 
\end{defn}
\begin{enumerate}
\item if $B_{p}$ $(p\in\underline{r})$ is a positive or negative block,
then for each $i\in B_{p}$ and $j\in B_{p}\bigcap\mathcal{N}_{i}$,
\[
\bar{w}_{ij}=\begin{cases}
w_{ij}, & \boldsymbol{v}_{\text{2}}(L)_{ij}>1\,\text{or}\,\boldsymbol{v}_{\text{2}}(L)_{ij}<0,\\
0, & 1\ge\boldsymbol{v}_{\text{2}}(L)_{ij}\ge0;
\end{cases}
\]
\item for all remaining $(i,j)\in\mathcal{E}$, $\bar{w}_{ij}=w_{ij}$.
\end{enumerate}
\noindent According to Definition \ref{def:fsn-tempo-network-autonomous},
the construction of the FSN network is built upon the block decomposition
of a graph; the edges in the core block and zero block will remain
unchanged while the other edges can be eliminated depending on the
quantity $\boldsymbol{v}_{\text{2}}(L)_{ij}$. We now proceed to examine
the reachability of the FSN network associated with FANs.
\begin{thm}
\label{thm:reachability-FSN-FAN} Let $\bar{\mathcal{G}}=(\bar{\mathcal{V}},\mathcal{\bar{E}},\bar{W})$
be the FSN network of the FAN (\ref{eq:consensus-overall}) on $\mathcal{G}=(\mathcal{V},\mathcal{E},W)$.
Then, the FAN (\ref{eq:consensus-overall}) achieves consensus on
the associated FSN network $\bar{\mathcal{G}}$. Moreover, the consensus
value is the average of the initial values of either the agents in
the core block and zero blocks or the core node and the agents in
zero blocks.
\end{thm}
\begin{proof}
Refer to the Appendix.
\end{proof}
According to the Theorem \ref{thm:reachability-FSN-FAN}, the FAN
(\ref{eq:consensus-overall}) can achieve consensus on the associated
FSN network $\bar{\mathcal{G}}$, however, the consensus value is
generally not equal to the consensus value achieved by the original
network (average of initial states of all agents). In fact, the consensus
value achieved on the FSN network $\bar{\mathcal{G}}$ is eventually
the average of the initial states of agents belonging to the core
block and the zero blocks or the average of the initial states of
the core node and agents belonging to zero blocks. We provide the
following example to demonstrate the reachability property of the
FSN network $\bar{\mathcal{G}}_{12}$ corresponding to the network
$\mathcal{G}_{12}$ in the left plot of Figure \ref{fig:12-node-network}.
\begin{example}
The FSN network $\bar{\mathcal{G}}_{12}$ corresponding to the network
$\mathcal{G}_{12}$ in Figure \ref{fig:12-node-network} is shown
in Figure \ref{fig:FSN-network of auto network}. The core block in
$\mathcal{G}_{12}$ is $B_{0}=\mathcal{G}\left(\left\{ 4,5,6\right\} \right)$.
As one can see from Figure \ref{fig:FSN-network of auto network},
all agents except that in the core block $B_{0}$ are reachable from
agents in the core block. 
\begin{figure}[tbh]
\begin{centering}
\begin{tikzpicture}[scale=0.7, >=stealth',   pos=.8,  photon/.style={decorate,decoration={snake,post length=1mm}} ]

	\node (n1) at (0.75,1) [circle,inner sep= 1pt,draw] {1};
	\node (n2) at (-0.5,0.5) [circle,inner sep= 1pt,draw] {2};
    \node (n3) at (2,0.5) [circle,inner sep= 1pt,draw] {3};
    \node (n4) at (0.6,-0.5) [circle,inner sep= 1pt,fill=black!20,draw]  {4};
	\node (n5) at (-0.9,-1.2) [circle,inner sep= 1pt,fill=black!20,draw]  {5};
	\node (n6) at (0.75,-1.75) [circle,inner sep= 1pt,fill=black!20,draw]  {6};
    \node (n7) at (-0.5,-2.5) [circle,inner sep= 1pt,draw] {7};
    \node (n8) at (1.1,-3) [circle,inner sep= 1pt,draw] {8};
    \node (n9) at (2,-2.3) [circle,inner sep= 1pt,draw] {9};
    \node (n10) at (2.2,-0.8) [circle,inner sep= 1pt,draw] {{\scriptsize 10}};
    \node (n11) at (0.55,2.3) [circle,inner sep= 1pt,draw] {{\scriptsize 11}};
    \node (n12) at (2,1.5) [circle,inner sep= 1pt,draw] {{\scriptsize 12}};

	\path[]
	(n2) [->,thick, right=8] edge node[below] {} (n1);
    \path[]
	(n3) [->,thick, right=8] edge node[below] {} (n1);

    \path[]
    (n4) [->,thick, right=8] edge node[below] {} (n1);
    \path[]
    (n1) [->,thick, right=8] edge node[below] {} (n11);
    \path[]
    (n1) [->,thick, right=8] edge node[below] {} (n12);

	\path[]
	(n4) [->,thick, right=8] edge node[below] {} (n2);
	\path[]
	(n4) [->,thick, right=8] edge node[below] {} (n3);

	\path[]
	(n4) [<->,thick, right=8] edge node[below] {} (n5);
    \path[]
    (n4) [<->,thick, right=8] edge node[below] {} (n6);
	\path[]
	(n5) [<->,thick, right=8] edge node[below] {} (n6);

	\path[]
	(n6) [->,thick, right=8] edge node[below] {} (n7);
    \path[]
	(n6) [->,thick, right=8] edge node[below] {} (n8);
    \path[]
    (n6) [->,thick, right=8] edge node[below] {} (n9);
    \path[]
    (n6) [->,thick, right=8] edge node[below] {} (n10);

\end{tikzpicture}
\par\end{centering}
\caption{The FSN network $\bar{\mathcal{G}}_{12}$ corresponding to the network
$\mathcal{G}_{12}$ in the left plot of Figure \ref{fig:12-node-network}.
The core block in $\mathcal{G}_{12}$ is\textcolor{red}{{} }$B_{0}=\mathcal{G}\left(\left\{ 4,5,6\right\} \right)$
which is highlighted in dark.}

\label{fig:FSN-network of auto network}
\end{figure}
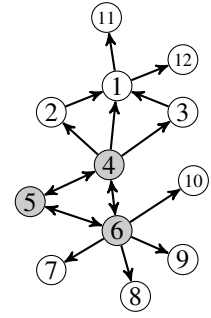
\end{example}
\begin{rem}[Generalized Monotonicity of Fiedler Vector]
 Remarkably,  Theorem \ref{thm:reachability-FSN-FAN} further extends
the Lemma \ref{lem:monotonicity-Fiedler} (a celebrated result by
Fiedler \cite{Fiedler1975}) by revealing the monotonicity property
of Fiedler's entries within each block, rather than only on cut nodes. 
\end{rem}

\subsection{Convergence Rate Enhancement }

We proceed to examine the convergence rate enhancement of FANs on
the corresponding FSN networks. First, we provide the following result,
for general FANs, that characterize the convergence rate of FAN (\ref{eq:consensus-overall})
on the associated FSN network.
\begin{prop}
\label{thm:convergence-rate-FSN-FAN} Let $\bar{\mathcal{G}}=(\mathcal{V},\mathcal{\bar{E}},\bar{W})$
be the FSN network of a FAN $\mathcal{G}=(\mathcal{V},\mathcal{E},W)$
characterized by (\ref{eq:consensus-overall}). Let $\lambda_{2}(L(\mathcal{\bar{G}}))$
and $\bar{\boldsymbol{v}}_{2}$ be the second smallest eigenvalue
of $L(\bar{\mathcal{G}})$ and the corresponding normalized eigenvector,
respectively. Then $\lambda_{2}(L(\mathcal{\bar{G}}))$ is lower bounded
by 
\begin{equation}
\lambda_{2}(L(\mathcal{G}))+{\displaystyle \sum_{(i,j)\in\mathcal{E}\setminus\mathcal{\bar{E}}}}[\bar{\boldsymbol{v}}_{2}]_{i}\left([\bar{\boldsymbol{v}}_{2}]_{j}-[\bar{\boldsymbol{v}}_{2}]_{i}\right).\label{eq:theorem-5}
\end{equation}
\end{prop}
\begin{proof}
Refer to the Appendix.
\end{proof}
One can observe from Proposition  \ref{thm:convergence-rate-FSN-FAN}
that the quantitative relationship between $\bar{\lambda}_{2}(L(\bar{\mathcal{G}}))$
and $\lambda_{2}(L(\mathcal{G}))$ is a bit vague since the second
term in (\ref{eq:theorem-5}) can be negative. This motivates us to
examine the special case of tree graphs- in which case, the convergence
rate between a FAN and the corresponding FSN network can be well-characterized.

\begin{thm}
\label{thm:FAN-Tree-convergence-rate}Let $\mathcal{T}$ be an $n$-node
tree network without zero blocks where $n\ge4$. Let $\bar{\mathcal{T}}$
be its FSN network characterized by (\ref{eq:consensus-overall}).
Let $\lambda_{2}(L(\bar{\mathcal{T}}))$ and $\lambda_{2}(L(\mathcal{T}))$
be the second smallest eigenvalue of $L(\bar{\mathcal{T}})$ and $L(\mathcal{T})$,
respectively. Then $\lambda_{2}(L(\bar{\mathcal{T}}))>\lambda_{2}(L(\mathcal{T}))$.
\end{thm}
\begin{proof}
Refer to the Appendix.
\end{proof}
\begin{example}
Consider a $12$-node FAN with the tree structure in Figure \ref{fig:tree-and-its-FSN}
(left) whose associated FSN network is shown in Figure \ref{fig:tree-and-its-FSN}
(right). The agents' initial states are $[\boldsymbol{x}(0)]_{1}=0.973$,
$[\boldsymbol{x}(0)]_{2}=0.649$, $[\boldsymbol{x}(0)]_{3}=0.8$,
$[\boldsymbol{x}(0)]_{4}=0.454$, $[\boldsymbol{x}(0)]_{5}=0.432$,
$[\boldsymbol{x}(0)]_{6}=0.825$, $[\boldsymbol{x}(0)]_{7}=0.084$,
$[\boldsymbol{x}(0)]_{8}=0.133$, $[\boldsymbol{x}(0)]_{9}=0.173$,
$[\boldsymbol{x}(0)]_{10}=0.391$, $[\boldsymbol{x}(0)]_{11}=0.831$,
$[\boldsymbol{x}(0)]_{12}=0.803$. 

In this example, the core block in tree $\mathcal{T}$ is the induced
subgraph $\mathcal{T}(\{4,6\})$. According to Figure \ref{fig:trajectory-12-node-tree},
the FAN (\ref{eq:consensus-overall}) on its associated FSN network
$\bar{\mathcal{G}}_{12}$ with the aforementioned initial states achieves
consensus on the value $0.6396$, which is equal to the average of
the initial states of the agents in the core block, namely, $\frac{1}{2}\left([\boldsymbol{x}(0)]_{4}+[\boldsymbol{x}(0)]_{6}\right)=0.6396.$
Computing $\boldsymbol{v}_{\text{2}}(L)$ corresponding to the Laplacian
matrix $L$ of $\mathcal{T}$ in this example yields, $[\boldsymbol{v}_{\text{2}}(L)]_{1}=0.333$,
$[\boldsymbol{v}_{\text{2}}(L)]_{2}=0.101$, $[\boldsymbol{v}_{\text{2}}(L)]_{3}=0.101$,
$[\boldsymbol{v}_{\text{2}}(L)]_{4}=0.079$, $[\boldsymbol{v}_{\text{2}}(L)]_{5}=0.101$,
$[\boldsymbol{v}_{\text{2}}(L)]_{6}=-0.257$, $[\boldsymbol{v}_{\text{2}}(L)]_{7}=-0.327$,
$[\boldsymbol{v}_{\text{2}}(L)]_{8}=-0.327$, $[\boldsymbol{v}_{\text{2}}(L)]_{9}=-0.327$,
$[\boldsymbol{v}_{\text{2}}(L)]_{10}=-0.327$, $[\boldsymbol{v}_{\text{2}}(L)]_{11}=0.424$,
$[\boldsymbol{v}_{\text{2}}(L)]_{12}=0.424$. Trajectories of $g_{ij}(t)$
in (\ref{eq:g_i}) for $i=1$ and $j\in\{3,5,11,12\}$ in the FAN
$\mathcal{T}$ shown in Figure \ref{fig:gij-tree}. One can see that
$g_{13}(t)\rightarrow\frac{[\boldsymbol{v}_{\text{2}}(L)]_{1}}{[\boldsymbol{v}_{\text{2}}(L)]_{3}}\approx3.299$
and $g_{1,11}(t)\rightarrow\frac{[\boldsymbol{v}_{\text{2}}(L)]_{1}}{[\boldsymbol{v}_{\text{2}}(L)]_{1}}\approx0.785$.
Moreover, $g_{15}(t)$ and $g_{1,12}(t)$ exhibit the same tendency
due to network symmetry. 

The convergence rate associated with networks $\mathcal{T}$ and $\mathcal{T}_{FSN}$
are $\lambda_{\text{2}}(L(\mathcal{T}))=0.2148$ and $\lambda_{\text{2}}(L(\mathcal{T}_{FSN}))=1$,
respectively. 
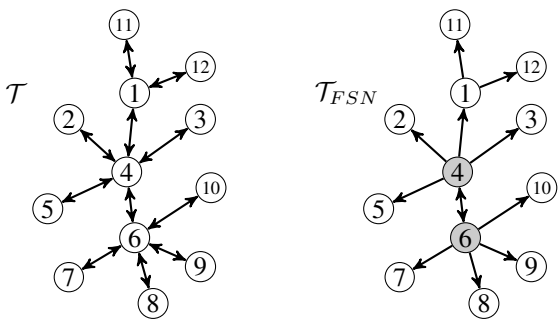
\begin{figure}[tbh]
\begin{centering}
\begin{tikzpicture}[scale=0.7, >=stealth',   pos=.8,  photon/.style={decorate,decoration={snake,post length=1mm}} ]	

   	\node (n1) at (0.75,1) [circle,inner sep= 1pt,draw] {1};
	\node (n2) at (-0.5,0.5) [circle,inner sep= 1pt,draw] {2};
    \node (n3) at (2,0.5) [circle,inner sep= 1pt,draw] {3};
    \node (n4) at (0.6,-0.5) [circle,inner sep= 1pt,draw] {4};
	\node (n5) at (-0.9,-1.2) [circle,inner sep= 1pt,draw] {5};
	\node (n6) at (0.75,-1.75) [circle,inner sep= 1pt,draw] {6};
    \node (n7) at (-0.5,-2.5) [circle,inner sep= 1pt,draw] {7};
    \node (n8) at (1.1,-3) [circle,inner sep= 1pt,draw] {8};
    \node (n9) at (2,-2.3) [circle,inner sep= 1pt,draw] {9};
    \node (n10) at (2.2,-0.8) [circle,inner sep= 1pt,draw] {{\scriptsize 10}};
    \node (n11) at (0.55,2.3) [circle,inner sep= 1pt,draw] {{\scriptsize 11}};
    \node (n12) at (2,1.5) [circle,inner sep= 1pt,draw] {{\scriptsize 12}};

    \node (G) at (-1.5,1.0) {$\mathcal{T}$};

    \path[]
    (n4) [<->,thick, right=8] edge node[below] {} (n1);
    \path[]
    (n1) [<->,thick, right=8] edge node[below] {} (n11);
    \path[]
    (n1) [<->,thick, right=8] edge node[below] {} (n12);

	\path[]
	(n4) [<->,thick, right=8] edge node[below] {} (n2);
	\path[]
	(n4) [<->,thick, right=8] edge node[below] {} (n3);

	\path[]
	(n4) [<->,thick, right=8] edge node[below] {} (n5);
    \path[]
    (n4) [<->,thick, right=8] edge node[below] {} (n6);

	\path[]
	(n6) [<->,thick, right=8] edge node[below] {} (n7);
    \path[]
	(n6) [<->,thick, right=8] edge node[below] {} (n8);
    \path[]
    (n6) [<->,thick, right=8] edge node[below] {} (n9);
    \path[]
    (n6) [<->,thick, right=8] edge node[below] {} (n10);
\end{tikzpicture}\,\,\,\,\,\,\,\,\,\,\,\,\,\,\,\,\,\,\begin{tikzpicture}[scale=0.7, >=stealth',   pos=.8,  photon/.style={decorate,decoration={snake,post length=1mm}} ]

	\node (n1) at (0.75,1) [circle,inner sep= 1pt,draw] {1};
	\node (n2) at (-0.5,0.5) [circle,inner sep= 1pt,draw] {2};
    \node (n3) at (2,0.5) [circle,inner sep= 1pt,draw] {3};
    \node (n4) at (0.6,-0.5) [circle,inner sep= 1pt,fill=black!20,draw] {4};
	\node (n5) at (-0.9,-1.2) [circle,inner sep= 1pt,draw] {5};
	\node (n6) at (0.75,-1.75) [circle,inner sep= 1pt,fill=black!20,draw] {6};
    \node (n7) at (-0.5,-2.5) [circle,inner sep= 1pt,draw] {7};
    \node (n8) at (1.1,-3) [circle,inner sep= 1pt,draw] {8};
    \node (n9) at (2,-2.3) [circle,inner sep= 1pt,draw] {9};
    \node (n10) at (2.2,-0.8) [circle,inner sep= 1pt,draw] {{\scriptsize 10}};
    \node (n11) at (0.55,2.3) [circle,inner sep= 1pt,draw] {{\scriptsize 11}};
    \node (n12) at (2,1.5) [circle,inner sep= 1pt,draw] {{\scriptsize 12}};

    \node (G) at (-1.5,1.0) {$\mathcal{T}_{FSN}$};

	\path[]
    (n4) [->,thick, right=8] edge node[below] {} (n1);
    \path[]
    (n1) [->,thick, right=8] edge node[below] {} (n11);
    \path[]
    (n1) [->,thick, right=8] edge node[below] {} (n12);

	\path[]
	(n4) [->,thick, right=8] edge node[below] {} (n2);
	\path[]
	(n4) [->,thick, right=8] edge node[below] {} (n3);

	\path[]
	(n4) [->,thick, right=8] edge node[below] {} (n5);
    \path[]
    (n4) [<->,thick, right=8] edge node[below] {} (n6);

	\path[]
	(n6) [->,thick, right=8] edge node[below] {} (n7);
    \path[]
	(n6) [->,thick, right=8] edge node[below] {} (n8);
    \path[]
    (n6) [->,thick, right=8] edge node[below] {} (n9);
    \path[]
    (n6) [->,thick, right=8] edge node[below] {} (n10);
\end{tikzpicture}
\par\end{centering}
\caption{A tree $\mathcal{T}$ (left) and its corresponding \textcolor{black}{FSN
network }$\mathcal{T}_{FSN}$ where \textcolor{black}{the} core block
$B=\left\{ 4,6\right\} $ is highlighted in dark\textcolor{black}{{}
(right).}}

\label{fig:tree-and-its-FSN}
\end{figure}
\begin{figure}[tbh]
\begin{centering}
\includegraphics[width=9cm]{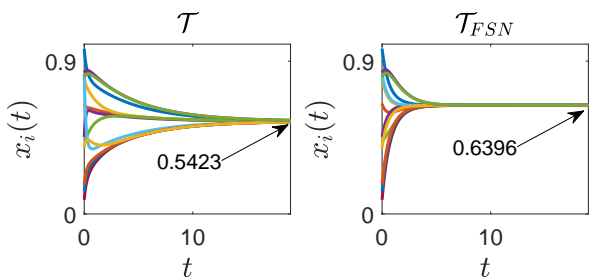}
\par\end{centering}
\caption{State trajectories of agents in FAN on the network $\mathcal{T}$
shown in the left plot of Figure \ref{fig:tree-and-its-FSN} as well
as its associated FSN network $\mathcal{T}_{FSN}$ in the right plot
of Figure \ref{fig:tree-and-its-FSN}.}
\label{fig:trajectory-12-node-tree}
\end{figure}
\begin{figure}[tbh]
\begin{centering}
\includegraphics[width=9cm]{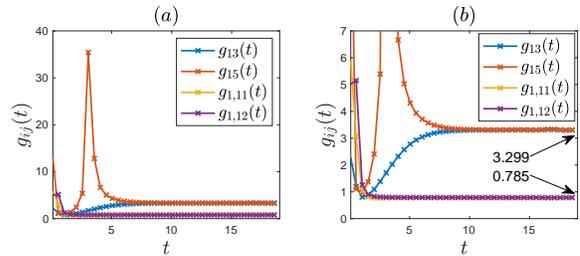}
\par\end{centering}
\caption{Trajectories of $g_{ij}(t)$ in (\ref{eq:g_i}) for  $i=1$ and $j\in\{3,5,11,12\}$
in the FAN $\mathcal{T}$ shown in Figure \ref{fig:tree-and-its-FSN}. }
\label{fig:gij-tree}
\end{figure}
\end{example}
\begin{rem}
Notably, the convergence rate of FSN network corresponding to trees
are always equal to one, however the convergence rate of a tree can
decrease dramatically when the diameter of the tree grows. For example,
that convergence rate on trees with diameter $\text{{\bf diam}}(\mathcal{T})$
is upper bounded by $2\left(1-\text{{\bf cos}}\left(\frac{\pi}{\text{{\bf diam}}(\mathcal{T})+1}\right)\right)$,
signifying that $\underset{\text{{\bf diam}}(\mathcal{T})\rightarrow\infty}{\text{\ensuremath{\lim}}}\lambda_{\text{2}}(L(\mathcal{T}))=0$.
\end{rem}

\subsection{Neighbor Selection in FANs}

We now examine networks whose second smallest eigenvalue of the Laplacian
matrix is simple. The following result establishes the relationship
between the relative tempo and the Fiedler vector of a network.
\begin{thm}
\label{thm:relative-tempo-FAN} Let $\mathcal{V}_{1}$ and $\mathcal{V}_{2}$
be two subsets of agents in $\mathcal{V}$. Each agent in $\mathcal{V}$
adopts dynamics (\ref{eq:consensus-protocol}). If the second smallest
eigenvalue of the Laplacian matrix $L$ is simple, then the relative
tempo of agents in $\mathcal{V}_{1}$ compared to that of $\mathcal{V}_{2}$
is \label{thm:relative-tempo}

\[
\mathbb{L}(\mathcal{V}_{1},\mathcal{V}_{2})=\frac{\|\phi(\mathcal{V}_{1})\boldsymbol{v}_{\text{2}}(L)\|}{\|\phi(\mathcal{V}_{2})\boldsymbol{v}_{\text{2}}(L)\|}.
\]
\end{thm}
\begin{proof}
According to Lemma \ref{lem:relative-tempo-symmetric-matrix} in Appendix,
the proof follows by choosing \textbf{$M=-L\otimes I_{d}$}.
\end{proof}
On the one hand, if the knowledge of both the core block and zero
block of a FAN is available, according to the Definition \ref{def:fsn-tempo-network-autonomous}
and Theorem \ref{thm:relative-tempo-FAN}, one can employ relative
tempo to construct the FSN network instead of using the information
in $\boldsymbol{v}_{\text{2}}(L)$. Under the FSN network, a consensus
at the value of the average of initial states associated with agents
in the core block and zero block, or the average of initial states
associated with agents in zero block and initial state of core node,
can be reached. In the meanwhile, the FSN network can be constructed
in a distributed manner, which is similar to SANs. 

On the other hand, if the knowledge of both the core block and zero
block of the network $\mathcal{G}$ is unavailable, a natural question
is whether one can determine the core block and zero block from the
network data. For general networks, this might be challenging. However,
such a data-driven approach can be adopted for tree networks, as every
node in a tree is a cut node; the core block in a tree network contains
at most two nodes. 

Here, we further discuss a class of the tree networks without zero
blocks and examine how to construct their FSN networks only using
local observations similar to the relative tempo. In fact, according
to the proof of Lemma \ref{lem:relative-tempo-symmetric-matrix} in
the Appendix, one can see that ${\displaystyle \lim_{t\rightarrow\infty}}\frac{\boldsymbol{e}_{1}^{\top}\dot{\boldsymbol{x}}_{u}(t)}{\boldsymbol{e}_{1}^{\top}\dot{\boldsymbol{x}}_{v}(t)}=\frac{[\boldsymbol{v}_{2}]_{u}}{[\boldsymbol{v}_{2}]_{v}}$,
where $u,v\in\mathcal{V}$. Note that for a tree network without zero
blocks, either there exists a core block containing two nodes $\{u,v\}$
such that $[\boldsymbol{v}_{2}]_{u}[\boldsymbol{v}_{2}]_{v}<0$, or
there exists a core node $\{w\}$ such that $[\boldsymbol{v}_{2}]_{w}=0$.
For the former case, one can use the quantity $\mathbb{L}^{\prime}(u,v)={\displaystyle \lim_{t\rightarrow\infty}}\frac{\boldsymbol{e}_{1}^{\top}\dot{\boldsymbol{x}}_{u}(t)}{\boldsymbol{e}_{1}^{\top}\dot{\boldsymbol{x}}_{v}(t)}$
instead of $\mathbb{L}(u,v)$ to identify nodes $u$ and $v$ in the
core block, then the edges between $u$ and $v$ are reserved while
the remaining edges shall be treated according to Definition \ref{def:fsn-tempo-network-autonomous}.
For the latter case, one can also use the quantity $\mathbb{L}^{\prime}(u,v)$
to construct the FSN network according to Definition \ref{def:fsn-tempo-network-autonomous}.

According to Definition \ref{def:fsn-tempo-network-autonomous} and
Theorem \ref{thm:relative-tempo-FAN}, the reduced neighbor set for
constructing FSN network of FAN on tree networks  can be defined as
follows.
\begin{defn}
Let $\mathcal{T}=(\mathcal{V},\mathcal{E},W)$ be a tree network without
zero blocks. The reduced neighbor set for $i\text{\ensuremath{\in\mathcal{V}}}$
to construct the associated FSN network of FAN (\ref{eq:consensus-overall})
is
\begin{eqnarray*}
\mathcal{N}_{i}^{\text{FSN}} & = & \left\{ j\in\mathcal{N}_{i}\thinspace\mid\thinspace\boldsymbol{v}_{\text{2}}(L(\mathcal{T}))_{ij}>1\,\text{or}\,\boldsymbol{v}_{\text{2}}(L(\mathcal{T}))_{ij}<0\right\} \\
 & = & \left\{ j\in\mathcal{N}_{i}\thinspace\mid\thinspace\mathbb{L}^{\prime}(i,j)>1\,\text{or}\,\mathbb{L}^{\prime}(i,j)<0\right\} .
\end{eqnarray*}
\end{defn}
To sum up, we have established the parallel framework of distributed
neighbor selection for FANs with a specific focus on tree networks.

\section{Extension to Signed Networks\label{sec:Extension-to-Signed}}

In this section, we discuss extensions of the aforementioned results
to structurally balanced signed networks. A signed network $\mathcal{G}=(\mathcal{V},\mathcal{E},W)$
is structurally balanced if there exists a bipartition of the node
set $\mathcal{V}$ (hence, $\mathcal{V}_{1}\subset\mathcal{V}$ and
$\mathcal{V}_{2}\subset\mathcal{V}$ such that $\mathcal{V}=\mathcal{V}_{1}\cup\mathcal{V}_{2}$
and $\mathcal{V}_{1}\cap\mathcal{V}_{2}=\emptyset$) such that the
edge weights within each subset are positive, but negative for edges
between the two subsets\ \cite{harary1953notion}. For structurally
balanced signed SANs, the eigenvectors of the perturbed Laplacian
can be transformed from the unsigned SANs via Gauge transformations
\cite{altafini2013consensus}. We now discuss the extension of our
results for structurally balanced signed networks. 

\subsection{Signed Semi-Autonomous Networks}

Consider the interaction protocol, 

\begin{align}
\dot{\boldsymbol{x}}_{i}(t) & =-\sum_{i=1}^{n}|w_{ij}|(\boldsymbol{x}_{i}(t)-\text{{\bf sgn}}(w_{ij})\boldsymbol{x}_{j}(t))\nonumber \\
 & -\sum_{l=1}^{m}|b_{il}|(\boldsymbol{x}_{i}(t)-\text{{\bf sgn}}(b_{il})\boldsymbol{u}_{l}),i\in\mathcal{V},\label{eq:LF-signed-protocol}
\end{align}
where $b_{il}\in\left\{ 1,-1\right\} $ if and only if $i\in\mathcal{V}_{\text{leader}}$
and $b_{il}=0$ otherwise. The sign function $\text{{\bf sgn}}(\cdot)$
is such that $\text{{\bf sgn}}(z)=1$ for $z>0$, $\text{{\bf sgn}}(z)=-1$
for $z<0$ and $\text{{\bf sgn}}(z)=0$ for $z=0$.

Denote the signed Laplacian matrix of $\mathcal{G}$ as $L^{s}=(l_{ij}^{s})\in\mathbb{R}^{n\times n}$,
where $l_{ii}^{s}=\sum{}_{j=1}^{n}|w_{ij}|$ for $i\in\mathcal{V}$
and $l_{ij}^{s}=-w_{ij}$ for $i\ne j$. The collective dynamics of
(\ref{eq:LF-signed-protocol}) is then,
\begin{equation}
\dot{\boldsymbol{x}}=-(L_{B}^{s}(\mathcal{G})\otimes I_{d})\boldsymbol{x}+(B\otimes I_{d})\boldsymbol{u},\label{eq:signed-semi-overall}
\end{equation}
where $L_{B}^{s}(\mathcal{G})=L^{s}+\text{{\bf diag}}(|B|\mathds{1}_{m})$,
$\boldsymbol{x}=(\boldsymbol{x}_{1}^{\top},\dots,\boldsymbol{x}_{n}^{\top})^{\top}\in\mathbb{R}^{dn}$,
$B=[b_{il}]\in\mathbb{R}^{n\times m}$ and $\boldsymbol{u}=(\boldsymbol{u}_{1}^{\top},\dots,\boldsymbol{u}_{m}^{\top})^{\top}\in\mathbb{R}^{dm}$.
Denote by the edge set between external inputs and the leaders and
the input set as $\mathcal{E}^{'}$ and $\mathcal{U}=(\boldsymbol{u}_{1},\dots,\boldsymbol{u}_{m})$,
respectively. The augmented graph $\widehat{\mathcal{G}}=(\widehat{\mathcal{V}},\widehat{\mathcal{E}},\widehat{W})$
is directed with $\widehat{\mathcal{V}}=\mathcal{V}\cup\mathcal{U}$,
$\widehat{\mathcal{E}}=\mathcal{E}\cup\mathcal{E}^{'}$ and $\widehat{W}=\left(\begin{array}{cc}
W & B\\
{\bf 0}_{m\times n} & {\bf 0}_{m\times m}
\end{array}\right)$. The signed Laplacian matrix of the network $\widehat{\mathcal{G}}$
is positive semi-definite if $\widehat{\mathcal{G}}$ is structurally
balanced \cite{altafini2013consensus}.
\begin{lem}
\label{thm:uniqueness-signed} Consider the signed SAN (\ref{eq:signed-semi-overall})
on a signed network $\mathcal{G}=(\mathcal{V},\mathcal{E},W)$. Suppose
that $\mathcal{G}=(\mathcal{V},\mathcal{E},W)$ is connected and $\widehat{\mathcal{G}}=(\widehat{\mathcal{V}},\widehat{\mathcal{E}},\widehat{W})$
is structurally balanced, and let $\lambda_{1}(L_{B}^{s})$ and $\boldsymbol{v}_{1}(L_{B}^{s})$
be the smallest eigenvalue of $L_{B}^{s}$ and the corresponding normalized
eigenvector, respectively. Then, $\lambda_{1}(L_{B}^{s})>0$ is a
simple eigenvalue of $L_{B}^{s}$ and $\boldsymbol{v}_{1}(L_{B}^{s})$
is positive under a proper Gauge transformation.
\end{lem}
\begin{proof}
The proof is an immediate extension of Lemma \ref{lem:uniqueness-unsigned}
and omitted for brevity.
\end{proof}
The FSN network for signed SANs can be defined as follows.
\begin{defn}[\textbf{FSN network of signed SANs}]
\label{def:fsn-network-SASN} Let $\mathcal{G}=(\mathcal{V},\mathcal{E},W)$
be a signed SAN characterized by (\ref{eq:signed-semi-overall}).
The FSN network of $\mathcal{G}$, denoted by $\bar{\mathcal{G}}=(\mathcal{\bar{V}},\bar{\mathcal{E}},\bar{W})$,
is a subgraph of $\mathcal{G}$ such that $\mathcal{\bar{V}}=\mathcal{V}$,
$\mathcal{\bar{E}}\subseteq\mathcal{E}$ and $\bar{W}=(\bar{w}_{ij})\in\mathbb{R}^{n\times n}$,
where $\bar{w}_{ij}=w_{ij}$ if $|\boldsymbol{v}_{\text{1}}(L_{B}^{s})_{ij}|>1$
and $\bar{w}_{ij}=0$ if $|\boldsymbol{v}_{\text{1}}(L_{B}^{s})_{ij}|\leq1$.
\end{defn}
\begin{thm}
\label{thm:reachability-FSN-signed}  Let $\bar{\mathcal{G}}=(\mathcal{V},\mathcal{\bar{E}},\bar{W})$
be the FSN network of the signed SAN $\mathcal{G}=(\mathcal{V},\mathcal{E},W)$
characterized by (\ref{eq:signed-semi-overall}). Suppose that $\mathcal{G}=(\mathcal{V},\mathcal{E},W)$
is connected and $\widehat{\mathcal{G}}=(\widehat{\mathcal{V}},\widehat{\mathcal{E}},\widehat{W})$
is structurally balanced; then all agents in $\bar{\mathcal{G}}$
are reachable from the external input. 
\end{thm}
\begin{proof}
Note that for a structurally balanced signed network, there exists
a quantitative connection between traditional Laplacian matrix and
signed Laplacian matrix via a Gauge transformation, assuming matrix
form $G=\text{{\bf diag}}\left\{ \sigma_{1},\cdots,\sigma_{n}\right\} $,
where $\sigma_{i}\in\{1,-1\}$ and $i\in\underline{n}$\ \cite{altafini2013consensus}.
The proof (omitted for brevity) follows from Lemma \ref{lem:steady-state-unsigned-semi},
proof of Theorem \ref{thm:reachability-FSN-SAN}, and applying the
Gauge transformation corresponding to the signed network $\mathcal{G}$. 
\end{proof}
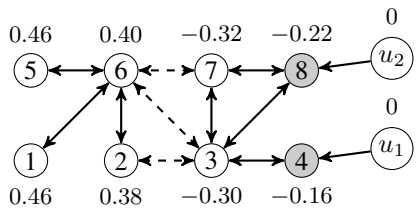
\begin{figure}[tbh]
\begin{centering}
\begin{tikzpicture}[scale=0.8, >=stealth',   pos=.8,  photon/.style={decorate,decoration={snake,post length=1mm}} ]
	\node (n1) at (0,0) [circle,inner sep= 1.5pt,draw] {1};
	\node (n2) at (1.5,0) [circle,inner sep= 1.5pt,draw] {2};
    \node (n3) at (3,0) [circle,inner sep= 1.5pt,draw] {3};
    \node (n4) at (4.5,0) [circle,inner sep= 1.5pt,fill=black!20,draw] {4};
	\node (n5) at (0,1.5) [circle,inner sep= 1.5pt,draw] {5};
	\node (n6) at (1.5,1.5) [circle,inner sep= 1.5pt,draw] {6};
    \node (n7) at (3,1.5) [circle,inner sep= 1.5pt,draw] {7};
    \node (n8) at (4.5,1.5) [circle,inner sep= 1.5pt,fill=black!20,draw] {8};

	\node (TV1) at (0,-0.6)   {{\small $0.46$}};
	\node (TV2) at (1.5,-0.6)   {{\small $0.38$}};
	\node (TV3) at (3,-0.6)   {{\small $-0.30$}};
	\node (TV4) at (4.5,-0.6)   {{\small $-0.16$}};
	\node (TV5) at (0,2.1)   {{\small $0.46$}};
	\node (TV6) at (1.5,2.1)   {{\small $0.40$}};
	\node (TV7) at (3,2.1)   {{\small $-0.32$}};
	\node (TV8) at (4.5,2.1)   {{\small $-0.22$}};

    \node (u1) at (6,0.2) [circle,inner sep= 1.5pt,draw] {$u_1$};
    \node (u2) at (6,1.7) [circle,inner sep= 1.5pt,draw] {$u_2$};

	\node (u11) at (6,2.4)   {{\small $0$}};
	\node (u22) at (6,0.9)   {{\small $0$}};

	\path[]
	(u1) [->,thick] edge node[below] {} (n4)	
    (u2) [->,thick] edge node[below] {} (n8);

	\path[]	(n2) [<->,dashed,thick] edge node[below] {} (n3);
	\path[]	(n6) [<->,thick] edge node[below] {} (n5);
	\path[] 	(n6) [<->,thick] edge node[below] {} (n1);
	\path[] 	(n2) [<->,thick] edge node[below] {} (n6); 
	\path[] 	(n7) [<->,dashed,thick] edge node[below] {} (n6); 
	\path[] 	(n3) [<->,dashed,thick] edge node[below] {} (n6);
	\path[] 	(n3) [<->,thick] edge node[below] {} (n7);
	\path[] 	(n8) [<->,thick] edge node[below] {} (n7);
	\path[] 	(n8) [<->,thick] edge node[below] {} (n3);
	\path[] 	(n4) [<->,thick] edge node[below] {} (n3);

\end{tikzpicture}
\par\end{centering}
\caption{An eight-node structurally balanced signed network $\mathcal{G}_{8}$
with two agents directly influenced by external inputs. The solid
lines and dashed lines represent the edges weighted by positive and
negative numbers, respectively. The entry in the $\boldsymbol{v}_{1}(L_{B}^{s})$
corresponding to each agent is shown close to each node.}

\label{fig:eight-node-consensus-network-SB-LF}
\end{figure}

\begin{thm}
\label{thm:relative-tempo-SAN-signed} Let $\mathcal{V}_{1}\subset\mathcal{V}$
and $\mathcal{V}_{2}\subset\mathcal{V}$ be two subsets of agents
in a connected signed SAN $\mathcal{G}=(\mathcal{V},\mathcal{E},W)$
characterized by (\ref{eq:signed-semi-overall}). If $\widehat{\mathcal{G}}=(\widehat{\mathcal{V}},\widehat{\mathcal{E}},\widehat{W})$
is structurally balanced, then the relative tempo between agents in
$\mathcal{V}_{1}$ and $\mathcal{V}_{2}$ satisfies ${\color{blue}\mathbb{L}}(\mathcal{V}_{1},\mathcal{V}_{2})=\frac{\|\phi(\mathcal{V}_{1})\boldsymbol{v}_{1}(L_{B}^{s})\|}{\|\phi(\mathcal{V}_{2})\boldsymbol{v}_{1}(L_{B}^{s})\|}.$
\end{thm}
\begin{proof}
According to Lemma \ref{lem:relative-tempo-symmetric-matrix}, the
proof follows by choosing $M=\left(\begin{array}{cc}
-L_{B}^{s}\otimes I_{d} & B\otimes I_{d}\\
{\bf 0}_{md\times nd} & {\bf 0}_{md\times md}
\end{array}\right)$ in Lemma \ref{lem:relative-tempo-symmetric-matrix} in Appendix. 
\end{proof}
\begin{thm}
\label{thm:convergence-rate-FSN-signed} Let $\bar{\mathcal{G}}=(\mathcal{V},\mathcal{\bar{E}},\bar{W})$
be the FSN network of a connected signed SAN $\mathcal{G}=(\mathcal{V},\mathcal{E},W)$
characterized by (\ref{eq:signed-semi-overall}). If $\widehat{\mathcal{G}}=(\widehat{\mathcal{V}},\widehat{\mathcal{E}},\widehat{W})$
is structurally balanced, then $\lambda_{1}(L_{B}^{s}(\bar{\mathcal{G}}))\ge\lambda_{1}(L_{B}^{s}(\mathcal{G})).$
\end{thm}
\begin{proof}
The proof is a straightforward extension of Theorem \ref{thm:convergence-rate-FSN-SAN},
and omitted for brevity.
\end{proof}
We now provide an example to illustrate the aforementioned results
on signed SANs.
\begin{figure}[tbh]
\centering{}\begin{tikzpicture}[scale=0.8, >=stealth',   pos=.8,  photon/.style={decorate,decoration={snake,post length=1mm}} ]
	\node (n1) at (0,0) [circle,inner sep= 1.5pt,draw] {1};
	\node (n2) at (1.5,0) [circle,inner sep= 1.5pt,draw] {2};
    \node (n3) at (3,0) [circle,inner sep= 1.5pt,draw] {3};
    \node (n4) at (4.5,0) [circle,inner sep= 1.5pt,fill=black!20,draw] {4};
	\node (n5) at (0,1.5) [circle,inner sep= 1.5pt,draw] {5};
	\node (n6) at (1.5,1.5) [circle,inner sep= 1.5pt,draw] {6};
    \node (n7) at (3,1.5) [circle,inner sep= 1.5pt,draw] {7};
    \node (n8) at (4.5,1.5) [circle,inner sep= 1.5pt,fill=black!20,draw] {8};

	\node (TV1) at (0,-0.6)   {{\small $0.46$}};
	\node (TV2) at (1.5,-0.6)   {{\small $0.38$}};
	\node (TV3) at (3,-0.6)   {{\small $-0.30$}};
	\node (TV4) at (4.5,-0.6)   {{\small $-0.16$}};
	\node (TV5) at (0,2.1)   {{\small $0.46$}};
	\node (TV6) at (1.5,2.1)   {{\small $0.40$}};
	\node (TV7) at (3,2.1)   {{\small $-0.32$}};
	\node (TV8) at (4.5,2.1)   {{\small $-0.22$}};

 \node (u1) at (6,0.2) [circle,inner sep= 1.5pt,draw] {$u_1$};     
    \node (u2) at (6,1.7) [circle,inner sep= 1.5pt,draw] {$u_2$};
	\node (u11) at (6,2.4)   {{\small $0$}}; 	
    \node (u22) at (6,0.9)   {{\small $0$}};
	\path[] 	
    (u1) [->,thick] edge node[below] {} (n4) 	
    (u2) [->,thick] edge node[below] {} (n8);
	\path[] 	(n3) [->,dashed,thick] edge node[below] {} (n2);
	\path[] 	(n6) [->,thick] edge node[below] {} (n5);
	\path[] 	(n6) [->,thick] edge node[below] {} (n1);
	\path[] 	(n2) [->,thick] edge node[below] {} (n6);
	\path[] 	(n7) [->,dashed,thick] edge node[below] {} (n6);
	\path[] 	(n3) [->,dashed,thick] edge node[below] {} (n6);
	\path[] 	(n3) [->,thick] edge node[below] {} (n7);
	\path[] 	(n8) [->,thick] edge node[below] {} (n7);
	\path[] 	(n8) [->,thick] edge node[below] {} (n3); 
	\path[] 	(n4) [->,thick] edge node[below] {} (n3);

\end{tikzpicture}\caption{FSN network of the structurally balanced signed network in Figure
\ref{fig:eight-node-consensus-network-SB-LF}.}
\label{fig:signed-SB-FSN}
\end{figure}
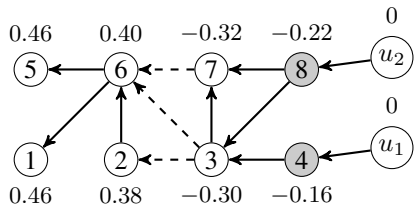

\begin{example}
Consider a signed SAN on the network $\mathcal{G}_{8}$ shown in Figure
\ref{fig:eight-node-consensus-network-SB-LF}, each agent holds a
three-dimensional state and agents $4$ and $8$ are leaders that
are directly influenced by the homogeneous input $\boldsymbol{u}=(\boldsymbol{u}_{1}^{\top},\thinspace\boldsymbol{u}_{2}^{\top})^{\top}$,
where $\boldsymbol{u}_{1}=\boldsymbol{u}_{2}=(0.7,\thinspace0.8,\thinspace0.9)^{\top}\in\mathbb{R}^{3}$.
The associated FSN network is shown in Figure \ref{fig:signed-SB-FSN}.
As one can see from Figure \ref{fig:convergence-rate-on-signed-SB},
the convergence rate of the bipartite consensus is significantly improved
on the associated FSN network.

\begin{figure}[tbh]
\begin{centering}
\includegraphics[width=9cm]{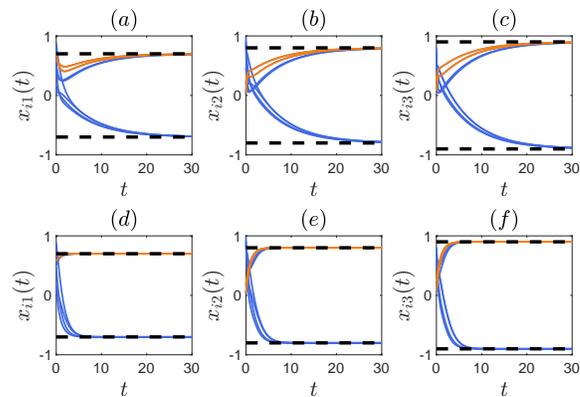}
\par\end{centering}
\caption{State trajectories of agents in the signed SAN (\ref{eq:signed-semi-overall})
on the structurally balanced signed network in Figure \ref{fig:eight-node-consensus-network-SB-LF}
((a)-(c)). State trajectories of agents in the signed SAN (\ref{eq:signed-semi-overall})
on FSN network in Figure \ref{fig:signed-SB-FSN} ((d)-(f)). The orange
and blue lines in each plot are trajectories of leader and follower
agents, respectively. The dotted lines in each plot represent external
inputs.}
\label{fig:convergence-rate-on-signed-SB}
\end{figure}

\end{example}

\subsection{Signed Fully-Autonomous Networks}

For the case of signed FANs, consider the interaction protocol, 

\begin{align}
\dot{\boldsymbol{x}}_{i}(t) & =-\sum_{i=1}^{n}|w_{ij}|(\boldsymbol{x}_{i}(t)-\text{{\bf sgn}}(w_{ij})\boldsymbol{x}_{j}(t)),i\in\mathcal{V},\label{eq:signed-protocol}
\end{align}
whose collective dynamics is
\begin{equation}
\dot{\boldsymbol{x}}=-(L^{s}(\mathcal{G})\otimes I_{d})\boldsymbol{x}.\label{eq:signed-semi-overall-1}
\end{equation}
Denote the unsigned network corresponding to the signed network $\mathcal{G}$
as $\widehat{\mathcal{G}}=(\mathcal{V},\mathcal{E},|W|)$, with Laplacian
matrix $L$. It is shown that $L=GL^{s}G$, where $G$ is the Gauge
transformation corresponding to the structurally balanced signed network
$\mathcal{G}$ \cite{altafini2013consensus}. This correspondence
implies that the Laplacian eigenvectors can be respectively transformed
via a Gauge transformation. Therefore, the results on the unsigned
FANs can be extended to the structurally balanced signed FANs through
a proper Gauge transformation on the Fielder vector.

\section{Acknowledgement}

The authors would like to thank the Associate Editor and the anonymous
reviewers for their constructive comments that improved the quality
of this manuscript.

\section{Conclusions Remarks \label{sec:Conclusion-Remarks}}

This paper addresses distributed neighbor selection problem of multi-agent
networks. In this direction, a theoretical framework of distributed
neighbor selection for diffusively coupled multi-agent networks has
been established. Along the way, we have highlighted the utility of
Laplacian eigenvectors to further improve network performance; these
eigenvectors encode hierarchical information about the network that
in turn, relate to the notion of relative tempo. The latter connection
is then used in a data-driven setting for the neighbor selection problem. 

Future works in this direction include extensions to directed and/or
time-varying networks, multi-agent systems with general individual
dynamics, and neighbor selection with noisy and delayed time-series
data. Furthermore, a notable feature of multi-agent networks is their
robustness to node/link failures. As such, it is often the case that
more links are favorable for the functionality of multi-agent systems,
e.g., convergence rate of the underlying coordination algorithm. Hence,
an interesting problem is to examine the optimal trade-off between
network robustness and size of the simplified network.

\appendix{}

The appendix contains the proofs of various results discussed in the
paper.

\section*{Proof of Lemma \ref{lem:uniqueness-unsigned}}
\begin{proof}
Note that the perturbed Laplacian matrix $L_{B}$ is symmetric and
diagonal dominant; as such, $\lambda_{i}(L_{B})\ge0$ for all $i\in\underline{n}$.
Assume that $L_{B}$ has an eigenvalue $\lambda_{1}(L_{B})=0$ with
associated eigenvector $\boldsymbol{v}_{1}(L_{B})\in\mathbb{R}^{n}$.
Then, 
\begin{align}
L_{B}\boldsymbol{v}_{1}(L_{B}) & =L\boldsymbol{v}_{1}(L_{B})+\text{{\bf diag}}(B\mathds{1}_{m})\boldsymbol{v}_{1}(L_{B})=0.
\end{align}
Multiply the above equality by $\boldsymbol{v}_{1}^{\top}(L_{B})$
from left yields,
\begin{equation}
\boldsymbol{v}_{1}^{\top}(L_{B})L\boldsymbol{v}_{1}(L_{B})+\boldsymbol{v}_{1}^{\top}(L_{B})\text{{\bf diag}}(B\mathds{1}_{m})\boldsymbol{v}_{1}(L_{B})=0.
\end{equation}
Since, $\boldsymbol{v}_{1}^{\top}(L_{B})L\boldsymbol{v}_{1}(L_{B})\ge0,$
and $\boldsymbol{v}_{1}^{\top}(L_{B})\text{{\bf diag}}(B\mathds{1}_{m})\boldsymbol{v}_{1}(L_{B})\ge0,$
one has, $\boldsymbol{v}_{1}^{\top}(L_{B})L\boldsymbol{v}_{1}(L_{B})=0,$
and $\boldsymbol{v}_{1}^{\top}(L_{B})\text{{\bf diag}}(B\mathds{1}_{m})\boldsymbol{v}_{1}(L_{B})=0.$
This however means that $\boldsymbol{v}_{1}(L_{B})=\mathds{1}_{n}$,
leading to having $\boldsymbol{v}_{1}^{\top}(L_{B})\text{{\bf diag}}(B\mathds{1}_{m})\boldsymbol{v}_{1}(L_{B})>0$.
This is a contradiction and therefore $\lambda_{i}>0$ for all $i\in\underline{n}$.

We shall proceed to show that $\lambda_{\text{1}}(L_{B})$ is simple
and $\boldsymbol{v}_{\text{1}}(L_{B})$ is positive. Denote by $L_{B}=\eta I-M$,
where $M\in\mathbb{R}^{n\times n}$ is a non-negative matrix and $\eta$
is the maximum value of the diagonal entries of $L_{B}$. Then $e^{-L_{B}}=e^{M-\eta I}=e^{-\eta I}e^{M}$.
Note that the matrix $M$ is non-negative, therefore, $e^{-L_{B}}$
is a non-negative matrix. In addition, since the network $\mathcal{G}$
is connected, $M$ is irreducible, implying that $e^{-L_{B}}$ is
a non-negative irreducible matrix. Thus, according to Perron--Frobenius
theorem for irreducible non-negative matrices, the eigenvalue $e^{-\lambda_{1}(L_{B})}$
is simple and the corresponding eigenvector $\boldsymbol{v}_{1}(L_{B})$
is positive.
\end{proof}

\section*{Steady-state of SAN on Unsigned Networks}

In the case of unsigned networks, the steady-state of SAN (\ref{eq:unsigned-LF-overall})
is consensus (when the external input is homogeneous) or cluster consensus
(when the external input is heterogeneous) \cite{cao2012distributed}.
Formally, the steady-state of the SAN (\ref{eq:unsigned-LF-overall})
is determined by the convex hull spanned by external inputs, namely,
 $\text{{\bf Co}}(\mathcal{U})=\{\sum_{i=1}^{m}k_{i}\boldsymbol{u}_{i}\thinspace|\thinspace\boldsymbol{u}_{i}\in\mathcal{U},\thinspace k_{i}\ge0,\sum_{i=1}^{m}k_{i}=1\}.$
Then following lemma characterizes the steady-state  of the SAN (\ref{eq:unsigned-LF-overall})
on unsigned networks.
\begin{lem}
\label{lem:steady-state-unsigned-semi}\textbf{ }\cite{lin2005necessary,ren2005consensus,cao2012distributed}\textbf{
}Consider the SAN (\ref{eq:unsigned-LF-overall}) on an unsigned network
$\mathcal{G}=(\mathcal{V},\mathcal{E},W)$. Then, the state of all
agents converge to the convex hull spanned by the external inputs
for arbitrary initial conditions if and only if for each agent $i\in\mathcal{V}$,
there exists at least one external input $\boldsymbol{u}_{l}\in\mathcal{U}$
such that $i$ is reachable from $\boldsymbol{u}_{l}$. Moreover,
the steady-state  of the SAN (\ref{eq:unsigned-LF-overall}) admits,
$\lim_{t\rightarrow\infty}\boldsymbol{x}(t)=(L_{B}^{-1}\otimes I_{d})(B\otimes I_{d})\boldsymbol{u}=(L_{B}^{-1}B)\otimes I_{d}\boldsymbol{u}.$
Specifically, if $\boldsymbol{u}$ is homogeneous, then the SAN (\ref{eq:unsigned-LF-overall})
achieves consensus, namely, $\lim_{t\rightarrow\infty}\boldsymbol{x}(t)=\frac{1}{m}\left({\color{blue}\mathds{1}_{n}}{\color{blue}\mathds{1}_{m}^{\top}}\otimes I_{d}\right)\boldsymbol{u}.$
\end{lem}

\section*{Proof of Theorem \ref{thm:relative-tempo-SAN}}

In order to show Theorem \ref{thm:relative-tempo-SAN}, we need the
following lemma.
\begin{lem}
\label{lem:relative-tempo-symmetric-matrix} Consider a matrix ordinary
differential equation  $\dot{\boldsymbol{x}}(t)=M\boldsymbol{x}(t),$
where $M\in\mathbb{R}^{n\times n}$ is symmetric and has $n$ linearly
independent eigenvectors and $\boldsymbol{x}(t)=(x_{1}(t),x_{2}(t),\dots,x_{n}(t))^{\top}$.
Denote the ordered eigenvalue of $M$ as $\lambda_{1}\le\lambda_{2}\le\cdots\le\lambda_{n}$
with associated mutually perpendicular  normalized eigenvectors $\boldsymbol{\varphi}_{1},\boldsymbol{\varphi}_{2},\ldots,\boldsymbol{\varphi}_{n}$.
Let $\lambda_{k_{1}}=\lambda_{k_{2}}=\ldots=\lambda_{k_{s}}$ be the
largest nonzero eigenvalue of $M$ with the algebraic multiplicity
$s\in\underline{n}$. Let $\psi(\eta_{1})=(\boldsymbol{e}_{i_{1}},\cdots,\boldsymbol{e}_{i_{s_{1}}})^{\top}\in\mathbb{R}^{s_{1}\times n}$
and $\psi(\eta_{2})=(\boldsymbol{e}_{j_{1}},\cdots,\boldsymbol{e}_{j_{s_{2}}})^{\top}\in\mathbb{R}^{s_{2}\times n}$
where $\eta_{1}=\left\{ i_{1},\ldots,i_{s_{1}}\right\} \subset\underline{n}$
and $\eta_{2}=\left\{ j_{1},\ldots,j_{s_{2}}\right\} \subset\underline{n}$,
respectively. Denote $\boldsymbol{\alpha}_{qi}=\psi(\eta_{q})\boldsymbol{\varphi}_{i}=\psi_{q}\boldsymbol{\varphi}_{i}\in\mathbb{R}^{s_{q}}$,
$S=[\boldsymbol{\varphi}_{1},\boldsymbol{\varphi}_{2},\ldots,\boldsymbol{\varphi}_{n}]\in\mathbb{R}^{n\times n}$
and $\boldsymbol{\beta}=[\beta_{1},\beta_{2},\ldots,\beta_{n}]^{\top}=S^{-1}\boldsymbol{x}(0)\in\mathbb{R}^{n}$
for $q\in\underline{2}$ and $i\in\underline{n}$. Then 
\begin{align}
\lim_{t\rightarrow\infty}\frac{\|\psi_{1}\dot{\boldsymbol{x}}(t)\|}{\|\psi_{2}\dot{\boldsymbol{x}}(t)\|} & =\left(\frac{{\displaystyle \sum_{i,j=1}^{s}}\lambda_{k_{i}}\lambda_{k_{j}}\boldsymbol{\alpha}_{1k_{i}}^{\top}\boldsymbol{\alpha}_{1k_{j}}\beta_{k_{i}}\beta_{k_{j}}}{{\displaystyle \sum_{i,j=1}^{s}}\lambda_{k_{i}}\lambda_{k_{j}}\boldsymbol{\alpha}_{2k_{i}}^{\top}\boldsymbol{\alpha}_{2k_{j}}\beta_{k_{i}}\beta_{k_{j}}}\right)^{\frac{1}{2}}.\label{eq:lemma-ratio-limit}
\end{align}
\end{lem}
\begin{proof}
Note that $M=SJS^{-1}$ where $S=[\boldsymbol{\varphi}_{1},\boldsymbol{\varphi}_{2},\ldots,\boldsymbol{\varphi}_{n}]\in\mathbb{R}^{n\times n}$
and $J=\text{{\bf diag}}\left\{ \lambda_{1},\lambda_{2},\ldots,\lambda_{n}\right\} \in\mathbb{R}^{n\times n}$.
According to the solution to the matrix ordinary differential equation
$\dot{\boldsymbol{x}}(t)=M\boldsymbol{x}(t)$, the derivative of $\boldsymbol{x}(t)$
is $\dot{\boldsymbol{x}}(t)=Me^{Mt}\boldsymbol{x}(0)=SJe^{Jt}S^{-1}\boldsymbol{x}(0)$.
Therefore one has,
\begin{align}
\|\psi_{q}\dot{\boldsymbol{x}}(t)\|^{2}= & \left(\dot{\boldsymbol{x}}(t)\right)^{\top}\psi_{q}^{\top}\psi_{q}\dot{\boldsymbol{x}}(t)\nonumber \\
= & \boldsymbol{x}(0)^{\top}(S^{-1})^{\top}e^{Jt}JS^{\top}\psi_{q}^{\top}\psi_{q}SJe^{Jt}S^{-1}\boldsymbol{x}(0)\nonumber \\
= & \sum_{i=1}^{n}\sum_{j=1}^{n}\lambda_{i}\lambda_{j}e^{(\lambda_{i}+\lambda_{j})t}\boldsymbol{\alpha}_{qi}^{\top}\boldsymbol{\alpha}_{qj}\beta_{i}\beta_{j}.
\end{align}
The statement of the lemma now follows from straightforward computation,
that has been omitted for brevity. 
\end{proof}
We are now in the position to prove Theorem \ref{thm:relative-tempo-SAN}.
\begin{proof}
Choose\textbf{ $M=\left(\begin{array}{cc}
-L_{B}\otimes I_{d} & B\otimes I_{d}\\
{\bf 0}_{md\times nd} & {\bf 0}_{md\times md}
\end{array}\right)$} in Lemma\textbf{ }\ref{lem:relative-tempo-symmetric-matrix}. Since
the algebraic multiplicity\textbf{ }of the largest nonzero eigenvalue
of the matrix $L_{B}$\textbf{ }is\textbf{ $1$},\textbf{ }the proof
then follows from a straightforward computation.
\end{proof}

\section*{Proof of Theorem \ref{thm:reachability-FSN-FAN}}

In order to prove Theorem \ref{thm:reachability-FSN-FAN}, we need
the following results.
\begin{lem}
\label{minimal edge bound}Let $\mathcal{G}(\mathcal{S})$ be a connected
induced subgraph of $\mathcal{G}(\mathcal{V})$ whose node set is
$\mathcal{S}\in\mathcal{V}$. Denote $\mathcal{E}^{bound}=\left\{ (i,j)\in\mathcal{E}\thinspace|\thinspace i\in\mathcal{S}\,\text{and}\,j\in\mathcal{V}\setminus\mathcal{S}\right\} $
and $\mathcal{V}_{\mathcal{S}}^{out-bound}=\left\{ j\in\mathcal{V}\setminus\mathcal{S}\thinspace|\thinspace\exists\thinspace i\in\mathcal{S}\,\text{such that}\,(i,j)\in\mathcal{E}^{bound}\right\} $
If for all $i\in\mathcal{S}$ and $j\in\mathcal{V}_{\mathcal{S}}^{out-bound}$,
$[\boldsymbol{v}_{2}]_{i}$ and $[\boldsymbol{v}_{2}]_{j}$ have the
same signs, then there exists an edge $(i,j)\in\mathcal{E}^{bound}$
such that $|[\boldsymbol{v}_{2}]_{i}|>|[\boldsymbol{v}_{2}]_{j}|$
.
\end{lem}
\begin{proof}
Assume that $[\boldsymbol{v}_{2}]_{i}$ and $[\boldsymbol{v}_{2}]_{j}$
are positive for all $i\in\mathcal{S}$ and $j\in\mathcal{V}_{\mathcal{S}}^{out-bound}$
and there does not exist an edge $(i,j)\in\mathcal{E}^{bound}$ such
that $[\boldsymbol{v}_{2}]_{i}>[\boldsymbol{v}_{2}]_{j}$. Let $\lambda_{2}$
be the second smallest eigenvalue of $L(\mathcal{G})$ with the corresponding
eigenvector $\boldsymbol{v}_{2}=([\boldsymbol{v}_{2}]_{1},[\boldsymbol{v}_{2}]_{2},\cdots,[\boldsymbol{v}_{2}]_{n})^{\top}\in\mathbb{R}^{n}$.
Denote the agents in $\mathcal{S}$ as $\{i_{1},i_{2},\ldots,i_{p}\}$;
then examining all the $i_{k}$th row in eigen-equation $L(\mathcal{G})\boldsymbol{v}_{2}=\lambda_{2}\boldsymbol{v}_{2}$,
yields,

\begin{equation}
\left({\displaystyle \sum_{j\in\mathcal{N}_{i_{k}}}}w_{i_{k}j}\right)[\boldsymbol{v}_{2}]_{i_{k}}-{\displaystyle \sum_{j\in\mathcal{N}_{i_{k}}}}w_{i_{k}j}[\boldsymbol{v}_{2}]_{j}=\lambda_{2}[\boldsymbol{v}_{2}]_{i_{k}},
\end{equation}
for all $k\in\underline{p}.$ Thereby,
\begin{equation}
\sum_{(i,j)\in\mathcal{E}^{bound}}w_{ij}([\boldsymbol{v}_{2}]_{i}-[\boldsymbol{v}_{2}]_{j})=\lambda_{2}\left({\displaystyle \sum_{i\in\mathcal{S}}}[\boldsymbol{v}_{2}]_{i}\right).
\end{equation}
Since $[\boldsymbol{v}_{2}]_{i}-[\boldsymbol{v}_{2}]_{j}\leq0,\forall(i,j)\in\mathcal{E}^{bound},$
one can conclude that $\lambda_{2}\leq0$, which is a contradiction
given the fact that $\lambda_{2}>0$ and there exists an edge $(i,j)\in\mathcal{E}^{bound}$
such that $[\boldsymbol{v}_{2}]_{i}>[\boldsymbol{v}_{2}]_{j}$.

For the case that $[\boldsymbol{v}_{2}]_{i}$ and $[\boldsymbol{v}_{2}]_{j}$
are negative for all $i\in\mathcal{S}$ and $j\in\mathcal{V}_{\mathcal{S}}^{out-bound}$,
the proof is analogous.
\end{proof}
\begin{lem}
\cite[Theorem 3.3]{Fiedler1975} \label{lem:Connectivity}Let $\mathcal{G}=(\mathcal{V},\mathcal{E},W)$
be an unsigned network with associated Laplacian matrix $L$. Let
$\boldsymbol{v}_{\text{2}}$ denote the eigenvector corresponding
to the second smallest eigenvalue of $L$. For any $r_{1}\geq0$ and
$r_{2}\leq0$, let $M(r_{1})=\left\{ i\in\mathcal{V}\thinspace|\thinspace[\boldsymbol{v}_{2}]_{i}+r_{1}\geq0\right\} $
and $M(r_{2})=\left\{ i\in\mathcal{V}\thinspace|\thinspace[\boldsymbol{v}_{2}]_{i}+r_{2}\leq0\right\} .$
Then, the subgraph $\mathcal{G}(r_{1})$ ($\mathcal{G}(r_{2})$) induced
by $\mathcal{G}$ on $M(r_{1})$ ($M(r_{2})$) is connected.
\end{lem}
We are now ready to prove Theorem \ref{thm:reachability-FSN-FAN}.
\begin{proof}
According to Lemma \ref{lem:steady-state-unsigned-semi}, consensus
can be examined as a reachability problem; we shall therefore consider
the reachability aspect first. According to Lemma \ref{lem:monotonicity-Fiedler},
we need to discuss two cases. 

Case 1: The network $\mathcal{G}$ contains a core block $B_{0}$,
we shall discuss the reachability of the positive/negative blocks
from $B_{0}$. Without the loss of generality, let $B_{i}$ be a positive
block, then two possible situations may occur in terms of the location
of $B_{i}$. 

First, consider the case where $B_{i}$ is a positive block connecting
with the core block $B_{0}$ directly through the cut node $i^{*}$.
We shall prove that for any node $j\in B_{i}$, $[\boldsymbol{v}_{2}]_{j}\geq[\boldsymbol{v}_{2}]_{i^{*}}$.
By contradiction, assume that there exists a node $j_{0}\in B_{i}$
satisfying $[\boldsymbol{v}_{2}]_{j_{0}}<[\boldsymbol{v}_{2}]_{i^{*}}$.
As $i^{*}$ is a cut node, according to Lemma \ref{lem:Connectivity},
for any node $p\in B_{0}$, one has $[\boldsymbol{v}_{2}]_{p}\geq[\boldsymbol{v}_{2}]_{i^{*}}$;
this is a contradiction (with the property of $B_{0}$). Therefore,
one has $[\boldsymbol{v}_{2}]_{j}\geq[\boldsymbol{v}_{2}]_{i^{*}}$
for any node $j\in B_{i}$. 

Second, consider the case where  $B_{i}$ is a positive block such
that all its neighboring blocks are positive. Denote by $\mathcal{N}_{i}^{B}=\left\{ B_{i_{1}},B_{i_{2}},\ldots,B_{i_{q_{i}}}\right\} $
as the neighboring blocks of $B_{i}$ with the corresponding cut nodes
$\left\{ i_{1},i_{2},\ldots,i_{q_{i}}\right\} $. We denote  $i^{*}=\underset{k\in\underline{q_{i}}}{\text{{\bf argmin}}}[\boldsymbol{v}_{2}]_{i_{k}},$
which connects blocks $B_{i}$ and $B_{i^{*}}$. Then we show that
for any node $j\in B_{i}$, $[\boldsymbol{v}_{2}]_{j}\geq[\boldsymbol{v}_{2}]_{i^{*}}$.
By contradiction, assume that there exists $j_{0}\in B_{i}$ satisfying
$[\boldsymbol{v}_{2}]_{j_{0}}<[\boldsymbol{v}_{2}]_{i^{*}}$; then
according to Lemma \ref{lem:Connectivity}, for any node $p\in B_{i^{*}}$,
one has $[\boldsymbol{v}_{2}]_{p}\ge[\boldsymbol{v}_{2}]_{i^{*}}$.
Otherwise, assume that there exists a node $p_{0}\in B_{i^{*}}$ such
that $[\boldsymbol{v}_{2}]_{p_{0}}<[\boldsymbol{v}_{2}]_{i^{*}}$,
and choose $r=\text{{\bf min}}\{-[\boldsymbol{v}_{2}]_{j_{0}},-[\boldsymbol{v}_{2}]_{p_{0}}\}.$
Note that $i^{*}$ is a cut node; then the subgraph $\mathcal{G}(r)$
induced by $\mathcal{G}$ on $M(r)$ is disconnected. Therefore, one
has $[\boldsymbol{v}_{2}]_{p}\ge[\boldsymbol{v}_{2}]_{i^{*}}$ for
any node $p\in B_{i^{*}}$. In addition, for any $k\in\underline{q_{i}}$
and $k\neq i^{*}$, one has $[\boldsymbol{v}_{2}]_{q}\geq[\boldsymbol{v}_{2}]_{i_{k}}$
for any node $q\in B_{i_{k}}$. This is due to having $i_{k}$ as
a cut node and $[\boldsymbol{v}_{2}]_{i_{k}}>[\boldsymbol{v}_{2}]_{i^{*}}$.
Therefore, for any $k\in\underline{q_{i}}$, one has $[\boldsymbol{v}_{2}]_{q}\geq[\boldsymbol{v}_{2}]_{i_{k}}$
for any node $q\in B_{ik}$. However, in view of Lemma \ref{minimal edge bound},
this is a contradiction. Therefore, for any node $j\in B_{i}$, one
has $[\boldsymbol{v}_{2}]_{j}\geq[\boldsymbol{v}_{2}]_{i^{*}}$.

Based on the above two scenarios, in the following, let $B_{i}$ be
an arbitrary positive block and denote by $\mathcal{N}_{i}^{B}=\left\{ B_{i_{1}},B_{i_{2}},\ldots,B_{i_{q_{i}}}\right\} $
as the neighboring blocks of $B_{i}$ with the corresponding cut nodes
$\left\{ i_{1},i_{2},\ldots,i_{q_{i}}\right\} $. Let $i^{*}=\underset{k\in\underline{q_{i}}}{\text{{\bf argmin}}}[\boldsymbol{v}_{2}]_{i_{k}},$
connecting blocks $B_{i}$ and $B_{i^{*}}$. We shall prove that any
node $j\in B_{i}$ can be reached by the cut node $i^{*}$ in the
FSN network $\bar{\mathcal{G}}$. Let $\bar{\mathcal{G}}(\left\{ i_{1},\cdots,i_{s_{0}}\right\} )$
be weakly connected, not reachable from the cut node $i^{*}$ in $\bar{\mathcal{G}}$,
where $s_{0}\in\mathbb{Z}_{+}$. Denote $\mathcal{E}^{bound}=\left\{ (i,j)\in\mathcal{E}\thinspace|\thinspace i\in\left\{ i_{1},\cdots,i_{s_{0}}\right\} ,j\in\mathcal{V}\setminus\left\{ i_{1},\cdots,i_{s_{0}}\right\} \right\} $,
then according to Definition \ref{def:fsn-tempo-network-autonomous},
for any edge $(i,j)\in\mathcal{E}^{bound}$, one has $[\boldsymbol{v}_{2}]_{i}\leq[\boldsymbol{v}_{2}]_{j}$,
which is a contradiction in view of Lemma \ref{minimal edge bound}.
Hence, any node $j\in B_{i}$, can be reached from the cut node $i^{*}$
in $\bar{\mathcal{G}}$.

Case 2: For the case that the network $\mathcal{G}$ only contains
a core node $i_{0}$; let nodes $j$ and $k$ be arbitrary positive
and negative nodes connecting the core node directly, respectively.
By Definition \ref{def:fsn-tempo-network-autonomous}, nodes $j$
and $k$ are reachable from the core node $i_{0}$. For the remaining
 nodes that do not connect the core node directly, the proof  is then
similar to the Case 1. 

Consequently, consensus can be guaranteed for aforementioned two cases
according to Lemma \ref{lem:steady-state-unsigned-semi}.

Secondly, we shall prove that the steady-state of agents in $\bar{\mathcal{G}}$
is equal to either the average of initial states of the agents in
the core block and zero blocks (Case 1) or that of the agents in zero
blocks and the core node (Case 2). To simplify our presentation, we
employ a general description for both cases, namely, that there are
$m$ agents in the union of either core block and zero blocks (Case
1) or core node and zero blocks (Case 2). Then, denote by $\mathcal{L}=\{1,2,\ldots,m\}$
and $\mathcal{F}=\{m+1,m+2,\ldots,n\}$, where $m\in\mathbb{Z}_{+}$
and $m\leq n$. We can represent the Laplacian matrix of $\mathcal{G}$
as, $L=\left[\begin{array}{cc}
L_{11} & 0_{m\times(n-m)}\\
L_{21} & L_{22}
\end{array}\right],$ where $L_{11}\in\mathbb{R}^{m\times m}$, $L_{22}\in\mathbb{R}^{(n-m)\times(n-m)}$
and $L_{21}\in\mathbb{R}^{(n-m)\times m}$, and $L_{22}$ is nonsingular.
Due to $L\mathds{1}_{n}=\boldsymbol{0}$, then one has $L_{21}\mathds{1}_{m}+L_{22}\mathds{1}_{n-m}=\boldsymbol{0}$.
Therefore, $L_{22}^{-1}L_{21}\mathds{1}_{m}=-\mathds{1}_{n-m}$. Denote
by $\boldsymbol{x}_{\mathcal{L}}=(\boldsymbol{x}_{1}^{\top}(t),\dots,\boldsymbol{x}_{m}^{\top}(t))^{\top}\in\mathbb{R}^{md}$
and $\boldsymbol{x}_{\mathcal{F}}=(\boldsymbol{x}_{m+1}^{\top}(t),\dots,\boldsymbol{x}_{n}^{\top}(t))^{\top}\in\mathbb{R}^{(n-m)d}$.
Thereby,
\begin{equation}
\dot{\boldsymbol{x}}_{\mathcal{L}}=-(L_{11}\otimes I_{d})\boldsymbol{x}_{\mathcal{L}},
\end{equation}
and
\begin{equation}
\dot{\boldsymbol{x}}_{\mathcal{F}}=-(L_{22}\otimes I_{d})\boldsymbol{x}_{\mathcal{F}}-(L_{21}\otimes I_{d})\boldsymbol{x}_{\mathcal{L}}.
\end{equation}

Let $X=\boldsymbol{x}_{\mathcal{F}}+(L_{22}^{-1}L_{21}\otimes I_{d})\boldsymbol{x}_{\mathcal{L}}.$
Then,
\begin{equation}
\dot{X}=-(L_{22}\otimes I_{d})X-(L_{22}^{-1}L_{21}L_{11}\otimes I_{d})\boldsymbol{x}_{\mathcal{L}}.
\end{equation}
According to the input-to-state stability theory, one has $\underset{t\rightarrow\infty}{\text{{\bf lim}}}X=\boldsymbol{0}$.
Therefore, 
\begin{align}
\underset{t\rightarrow\infty}{\text{{\bf lim}}}\boldsymbol{x}_{\mathcal{F}} & =-(L_{22}^{-1}L_{21}\otimes I_{d})\underset{t\rightarrow\infty}{\text{{\bf lim}}}\boldsymbol{x}_{\mathcal{L}}\\
 & =-\frac{1}{m}(L_{22}^{-1}L_{21}\mathds{1}_{m}\mathds{1}_{m}^{\top}\otimes I_{d})\boldsymbol{x}_{\mathcal{L}}(0)\\
 & =\frac{1}{m}(\mathds{1}_{n-m}\mathds{1}_{m}^{\top}\otimes I_{d})\boldsymbol{x}_{\mathcal{L}}(0),
\end{align}
which concludes the proof.
\end{proof}

\section*{Proof of Proposition \ref{thm:convergence-rate-FSN-FAN}}
\begin{proof}
Let $E=L(\bar{\mathcal{G}})-L(\mathcal{G})$. Note that $L(\bar{\mathcal{G}})\bar{\boldsymbol{v}}_{2}=\lambda_{2}(L(\bar{\mathcal{G}}))\bar{\boldsymbol{v}}_{2}$
and $\bar{\boldsymbol{v}}_{2}^{\top}L^{\top}(\bar{\mathcal{G}})=\lambda_{2}(L(\bar{\mathcal{G}}))\bar{\boldsymbol{v}}_{2}^{\top}.$
Hence, $\bar{\boldsymbol{v}}_{2}^{\top}L(\bar{\mathcal{G}})\bar{\boldsymbol{v}}_{2}=\lambda_{2}(L(\bar{\mathcal{G}}))\bar{\boldsymbol{v}}_{2}^{\top}\bar{\boldsymbol{v}}_{2},$
and $\bar{\boldsymbol{v}}_{2}^{T}L^{T}(\bar{\mathcal{G}})\bar{\boldsymbol{v}}_{2}=\lambda_{2}(L(\bar{\mathcal{G}}))\bar{\boldsymbol{v}}_{2}^{T}\bar{\boldsymbol{v}}_{2}.$
Therefore,
\begin{equation}
\bar{\boldsymbol{v}}_{2}^{\top}(L(\bar{\mathcal{G}})+L^{\top}(\bar{\mathcal{G}}))\bar{\boldsymbol{v}}_{2}=2\lambda_{2}(L(\bar{\mathcal{G}}))\bar{\boldsymbol{v}}_{2}^{\top}\bar{\boldsymbol{v}}_{2},
\end{equation}
and
\begin{align}
\lambda_{2}(L(\bar{\mathcal{G}}))= & \frac{\bar{\boldsymbol{v}}_{2}^{\top}(L(\bar{\mathcal{G}})+L^{\top}(\bar{\mathcal{G}}))\bar{\boldsymbol{v}}_{2}}{2\bar{\boldsymbol{v}}_{2}^{\top}\bar{\boldsymbol{v}}_{2}}\\
= & \frac{\bar{\boldsymbol{v}}_{2}^{\top}(L(\mathcal{G})+L^{\top}(\mathcal{G}))\bar{\boldsymbol{v}}_{2}}{2\bar{\boldsymbol{v}}_{2}^{\top}\bar{\boldsymbol{v}}_{2}}+\frac{\bar{\boldsymbol{v}}_{2}^{\top}(E+E^{\top})\bar{\boldsymbol{v}}_{2}}{2\bar{\boldsymbol{v}}_{2}^{\top}\bar{\boldsymbol{v}}_{2}}\\
= & \frac{\bar{\boldsymbol{v}}_{2}^{\top}L(\mathcal{G})\bar{\boldsymbol{v}}_{2}}{\bar{\boldsymbol{v}}_{2}^{\top}\bar{\boldsymbol{v}}_{2}}+\frac{\bar{\boldsymbol{v}}_{2}^{\top}(E+E^{\top})\bar{\boldsymbol{v}}_{2}}{2\bar{\boldsymbol{v}}_{2}^{\top}\bar{\boldsymbol{v}}_{2}}.
\end{align}
By applying Rayleigh theorem \cite[Theorem 4.2.2,  p.235]{horn2012matrix},
one has,
\begin{equation}
\lambda_{2}(L(\bar{\mathcal{G}}))\geq\lambda_{2}(L(\mathcal{G}))+\frac{\bar{\boldsymbol{v}}_{2}^{\top}(E+E^{\top})\bar{\boldsymbol{v}}_{2}}{2\bar{\boldsymbol{v}}_{2}^{\top}\bar{\boldsymbol{v}}_{2}}.
\end{equation}
Moreover, according to the expansion,
\begin{align}
\frac{\bar{\boldsymbol{v}}_{2}^{\top}(E+E^{\top})\bar{\boldsymbol{v}}_{2}}{2\bar{\boldsymbol{v}}_{2}^{\top}\bar{\boldsymbol{v}}_{2}} & =\frac{1}{\bar{\boldsymbol{v}}_{2}^{\top}\bar{\boldsymbol{v}}_{2}}{\displaystyle \sum_{(i,j)\in\mathcal{E}\setminus\mathcal{\bar{E}}}}[\bar{\boldsymbol{v}}_{2}]_{i}\left([\bar{\boldsymbol{v}}_{2}]_{j}-[\bar{\boldsymbol{v}}_{2}]_{i}\right),
\end{align}
and $\bar{\boldsymbol{v}}_{2}^{\top}\bar{\boldsymbol{v}}_{2}=1$,
one can conclude that,
\begin{equation}
\lambda_{2}(L(\bar{\mathcal{G}}))\geq\lambda_{2}(L(\mathcal{G}))+{\displaystyle \sum_{(i,j)\in\mathcal{E}\setminus\mathcal{\bar{E}}}}[\bar{\boldsymbol{v}}_{2}]_{i}\left([\bar{\boldsymbol{v}}_{2}]_{j}-[\bar{\boldsymbol{v}}_{2}]_{i}\right).
\end{equation}
\end{proof}

\section*{Proof of Theorem \ref{thm:FAN-Tree-convergence-rate}}
\begin{lem}
\cite{de2007old}\label{eigenvalue bound of tree graph} Let $\mathcal{T}$
be a non-star tree network with $n\geq4$ nodes. Then $\lambda_{2}(L(\mathcal{T}))<0.59$.
\end{lem}
We now prove Theorem \ref{thm:FAN-Tree-convergence-rate}.
\begin{proof}
Note that every node is a cut node in a tree network. Then in light
of Lemma \ref{lem:monotonicity-Fiedler}, there are two cases to consider.
For the case that there exists one core node in $\mathcal{T}$, denoted
by $w\in\mathcal{V}$, the Laplacian matrix $L(\bar{\mathcal{T}})$
of the FSN\textbf{ }network $\bar{\mathcal{T}}$ can be decomposed
as $L(\bar{\mathcal{T}})=\left[\begin{array}{cc}
0 & 0\\*
* & L_{\mathcal{V}\setminus\{w\}}
\end{array}\right],$ where $L_{\mathcal{V}\setminus\{w\}}$ is a lower triangular matrix,
with all diagonal elements equal to one. Therefore, the second smallest
eigenvalue of $L(\bar{\mathcal{T}})$ is equal to one. Hence, using
Lemma \ref{eigenvalue bound of tree graph},\textbf{ }it follows that
 \textbf{$\lambda_{2}(L(\bar{\mathcal{T}}))>\lambda_{2}(L(\mathcal{T}))$}.

Now we proceed to consider the case that there exists one core block
with two nodes; denote these nodes by $u\in\mathcal{V}$ and $v\in\mathcal{V}$.
Then, the Laplacian matrix $L(\bar{\mathcal{T}})$ of the FSN\textbf{
}network $\bar{\mathcal{T}}$ can be decomposed as, $L(\bar{\mathcal{T}})=\left[\begin{array}{cc}
L_{\{u,v\}} & 0\\*
* & L_{\mathcal{V}\setminus\{u,v\}}
\end{array}\right],$ where $L_{\{u,v\}}=\left[\begin{array}{cc}
1 & -1\\
-1 & 1
\end{array}\right]$ denotes the Laplacian matrix associated with the core block, and
$L_{\mathcal{V}\setminus\{u,v\}}$ is a lower triangular matrix with
all diagonal elements equal to one. Therefore, the second smallest
eigenvalue of $L(\bar{\mathcal{T}})$ is equal to one. Again, in view
of the Lemma \ref{eigenvalue bound of tree graph}\textbf{, }we conclude
that  \textbf{$\lambda_{2}(L(\bar{\mathcal{T}}))>\lambda_{2}(L(\mathcal{T}))$}.
\end{proof}
\bibliographystyle{IEEEtran}
\bibliography{mybib}

\end{document}